\documentclass[12pt]{article}

\usepackage[T1]{fontenc}
\usepackage[utf8]{inputenc}
\usepackage{amsmath,amsfonts,amsthm,amssymb}
\usepackage{booktabs,tabularx}
\usepackage{float}
\usepackage{xcolor}
\usepackage{geometry}
\usepackage{enumerate}
\usepackage{graphicx,color}
\usepackage{hyperref}
\usepackage{mathtools}
\usepackage{setspace}
\usepackage{cite}

\makeatletter
\renewcommand*\env@matrix[1][\arraystretch]{%
  \edef\arraystretch{#1}%
  \hskip -\arraycolsep
  \let\@ifnextchar\new@ifnextchar
  \array{*\c@MaxMatrixCols c}}
\makeatother

\newtheorem{theorem}{Theorem}[section]
\newtheorem{proposition}[theorem]{Proposition}
\newtheorem{corollary}[theorem]{Corollary}
\newtheorem{lemma}[theorem]{Lemma}

\newtheorem{definition}{Definition}[section]
\newtheorem{remark}[definition]{Remark}

\newcommand{\END}{\hfill\mbox{\raggedright $\Diamond$}\medskip} 
\newcommand{\etal}{\textrm{et al.}}
\newcommand{\ignore}[1]{}

\title{Homeostatic Mechanisms in Biological Systems}

\author{Pedro P. A. Cardoso de Andrade \\
Universidade de S\~ao Paulo \\
S\~ao Paulo, 05508-090, Brazil \\
\href{mailto:pedrop.and@gmail.com}%
{pedrop.and@gmail.com} 
\and 
Jo\~ao L. O. Madeira \\
University of Bath \\
Bath, BA2 7AY, UK \\
\href{mailto:joaoluizoliveiramadeira@gmail.com}%
{jldom20@bath.ac.uk}
\and 
Fernando Antoneli \\
Universidade Federal de S\~ao Paulo \\
S\~ao Paulo, 04039-032, Brazil \\
\href{mailto:fernando.antoneli@unifesp.br}%
{fernando.antoneli@unifesp.br} 
}

\date{\today}

\begin{document}

\maketitle

\begin{abstract}
In this paper we investigate the homeostatic mechanism in two biologically motivated models: intracellular copper regulation and self immune recognition. The analysis is based on the notions of infinitesimal homeostasis and near-perfect homeostasis. We introduce a theoretical background that makes it possible to consider points of infinitesimal homeostasis that lie at the boundary of the domain of definition of the input-output function. We show that the two models display near-perfect homeostasis. Moreover, we show that, unlike the examples of [M. Reed, J. Best, M. Golubitsky, I. Stewart, and H. F. Nijhout. Analysis of homeostatic mechanisms in biochemical networks. Bull. Math. Biol., 79(11):2534–2557, 2017], the obstruction of occurrence of infinitesimal homeostasis in both of them is caused by the modeling assumptions that force the point of infinitesimal homeostasis to lie at the boundary of domain of definition of the respective input-output functions.

\medskip

\noindent
{\bf Keywords:} Homeostasis, Input-Output Networks, Perfect Adaptation

\end{abstract}


\newpage

\section{Introduction}

A system exhibits \emph{homeostasis} if on change of an input variable $\mathcal{I}$ some observable $x_o(\mathcal{I})$ remains approximately constant.
Many researchers have emphasised that homeostasis is an important phenomenon in biology.  
For example, the extensive work of Nijhout, Reed, Best and collaborators~\cite{nrbu04,best09,nr14,nbr15,nbr18} consider biochemical networks associated with metabolic signalling pathways.
Further examples include regulation of cell number and size \cite{lloyd13}, control of sleep~\cite{wyatt99}, and expression level regulation in housekeeping genes~\cite{antoneli18}.  

Consider a dynamical system with input parameter $\mathcal{I}$ which varies over an open interval $]\mathcal{I}_{1}, \mathcal{I}_{2}[$. Suppose there is an output variable $x_o$ such that for each $\mathcal{I} \in ]\mathcal{I}_{1}, \mathcal{I}_{2}[$, the value $x_{o}(\mathcal{I})$ well-defined. In this situation, it is reasonable to say that the system would exhibit \emph{homeostasis} if after changing the input variable $\mathcal{I}$, the value of the observable $x_o(\mathcal{I})$ remains approximately constant. There are two formulations often considered by researchers: (1) the strict condition of \emph{perfect homeostasis}, where the observable $x_o(\mathcal{I})$ is required to be constant over a range of external stimuli; (2) the more general condition of \emph{near-perfect homeostasis}, where the observable $x_o(\mathcal{I})$ is required to be within a narrow interval of values over a range of external stimuli.

Golubitsky and Stewart \cite{gs17} proposed to employ methods from singularity theory to define the notion of \emph{infinitesimal homeostasis}. According to this approach, a system exhibits \emph{infinitesimal homeostasis} if $\frac{dx_{o}}{d \mathcal{I}}(\mathcal{I}_{0}) = 0$ for some input value $\mathcal{I}_{0}$, where $x_{o}$ is the function that associates to each input parameter $\mathcal{I}_{0}$ a unique value of the observable $x_o$, called \emph{input-output function}.

Reed \etal \cite{reed17} analyzed four distinct homeostatic mechanisms: feedforward excitation, feedback product inhibition, the kinetic motif, and the parallel inhibition motif. All of them occur in folate and methionine metabolism. Interestingly, \cite{reed17} showed that two of the motifs exhibit infinitesimal homeostasis and that although the other two do not, they all exhibit near-perfect homeostasis (see also \cite{yangyang20}).

\emph{Feedback product inhibition} is probably one of the simplest and best known homeostatic mechanisms in biochemistry. In its simplest form, product inhibition means that the product of a biochemical chain inhibits one or more of the enzymes involved in its own synthesis.
The differential equations are given by
\begin{equation} \label{feedback_inhibition}
\begin{aligned}
& \dot{x}_\iota = \mathcal{I} - g_1(x_{\iota}) - f(x_{\iota},x_o)  \\
& \dot{x}_\sigma = f(x_{\iota},x_o) - g_0(x_{\sigma}) - g_2(x_{\sigma})\\
& \dot{x}_o = g_0(x_{\sigma}) - g_3(x_o)
\end{aligned}
\end{equation}
where $f$, $g_i$ ($i=0,1,2,3$) are smooth functions.

Typically, the functions $g_i$ are defined on positive semi-axis, are linear and increasing.
The function $f$ is defined on the positive orthant and is positive.
The actual kinetic formulas for inhibitory function $f(x_{\iota},x_{o})$ have been extensively studied and depend on the details of the chemical binding of the substrate to one or more sites on the enzyme. One can impose general constraints on the function $f$ in order to get similar behavior: $\frac{\partial f}{\partial x_{\iota}} > 0$ (more substrate, faster reaction) and $\frac{\partial f}{\partial x_{o}} < 0$ (higher substrate, more inhibition of the reaction). 

Under these general conditions, it can be shown (see~\cite{yangyang20}) that the input-output function $x_o$ of \eqref{feedback_inhibition} is well-defined for all $\mathcal{I}>0$ and 
\[
 \frac{d x_o}{d \mathcal{I}} = 0 \qquad \Longleftrightarrow \qquad  
 \frac{\partial f}{\partial x_{\iota}} = 0
\]
That is, the assumption that $\frac{\partial f}{\partial x_{\iota}} > 0$ precludes occurrence of infinitesimal homeostasis. Moreover, it is shown in \cite{reed17} that near-perfect homeostasis is possible in such systems if one chooses an $f$ for which $\frac{\partial f}{\partial x_{\iota}} > 0$ is close to zero -- such a choice is consistent with the biochemistry of feedback product inhibition.

\ignore{
The jacobian matrix of \eqref{feedback_inhibition} at an equilibrium is given by
\begin{equation}
  J = \begin{pmatrix}
  -g_{1, x_{\iota}} - f_{x_{\iota}} & 0 & -f_{x_{o}} \\
  f_{x_{\iota}} & -g_{0,x_{\sigma}}-g_{2,x_{\sigma}} & f_{x_{o}} \\
  0 & g_{0,x_{\sigma}} & -g_{3, x_{o}}
  \end{pmatrix}
\end{equation}
with determinant given by
\[
 \det J = -(g_{1, x_{\iota}} + f_{x_{\iota}}) \, (g_{0,x_{\sigma}}+g_{2,x_{\sigma}}) \, g_{3, x_{o}}
 + f_{x_{o}} \, g_{0,x_{\sigma}} \, g_{1,x_{\sigma}}
\]
By assumption, $\det J<0$ at the equilibrium (because the equilibrium is attractor).

The homeostasis matrix is given by
\begin{equation}
  H = \begin{pmatrix}
  f_{x_{\iota}} & -g_{0,x_{\sigma}}-g_{2,x_{\sigma}} \\
  0 & g_{0,x_{\sigma}}
  \end{pmatrix}
\end{equation}
with homeostasis determinant $\det H = f_{x_{\iota}} \, g_{0,x_\sigma}$.
Since $g_{0,x_{\sigma}}$ is always positive,
\[
 x_o' = 0 \qquad \Longleftrightarrow \qquad f_{x_{\iota}} = 0
\]

Hence, the derivative of the input-output function $x_o$ is
\[
 x_o'=\frac{f_{x_{\iota}} \, g_{0,x_\sigma}}{(g_{1, x_{\iota}} + f_{x_{\iota}}) \, (g_{0,x_{\sigma}}+g_{2,x_{\sigma}}) \, g_{3, x_{o}}
 - f_{x_{o}} \, g_{0,x_{\sigma}} \, g_{1,x_{\sigma}}}
\]
}

The second example of \cite{reed17} exhibiting near-perfect homeostasis but not infinitesimal homeostasis is the \emph{parallel inhibition motif}. Again, the conclusion that infinitesimal homeostasis cannot occur in this system, follows from an incompatibility of a biochemical condition, called parallel inhibition hypotheses, and the condition that $\frac{d x_o}{d \mathcal{I}} = 0$.
Therefore, in both examples the obstruction to the occurrence of infinitesimal homeostasis comes from additional modeling assumptions due to the nature of the phenomena being modeled.

In this paper we consider another type of mechanism that may obstruct the occurrence of infinitesimal homeostasis. Namely, when the point of infinitesimal homeostasis is forced to be at the boundary of the domain of definition of the input-output function $x_o(\mathcal{I})$.

We introduce two biologically motivated models: intracellular copper regulation and self immune recognition. These two models can be represented by four node networks shown in Figure \ref{four_nodes_example_feedback}.
Interestingly, \cite{huang22} obtain the classification of ``homeostasis types'' in four-node core networks and the examples we consider here correspond to core equivalence classes 20 and 18 of \cite{huang22}, respectively.

In order to study the homeostatic mechanisms in those examples we first extend some of the theoretical results of \cite{wang20} to the case where the infinitesimal homeostasis point lies at the boundary. We introduce the notion of asymptotic infinitesimal homeostasis and show that the notion of core networks extend to this new situation.

\begin{figure}[!ht]
\centering
\begin{tabular}{c@{}c}
\includegraphics[scale = 0.35,trim=9cm 5cm 7cm 5cm,clip=true]%
{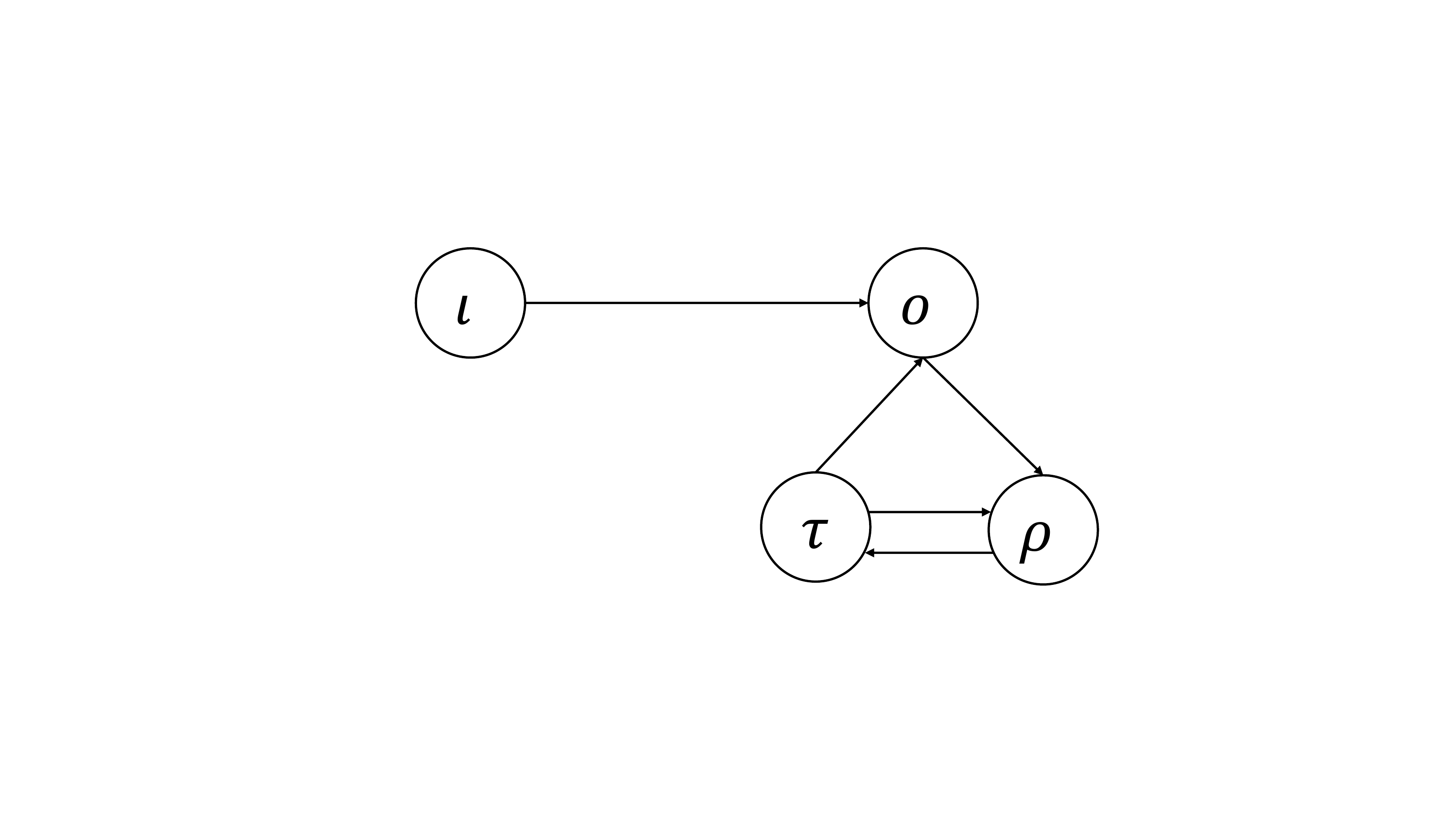} &
\includegraphics[scale = 0.33,trim=6cm 4cm 6cm 3cm,clip=true]%
{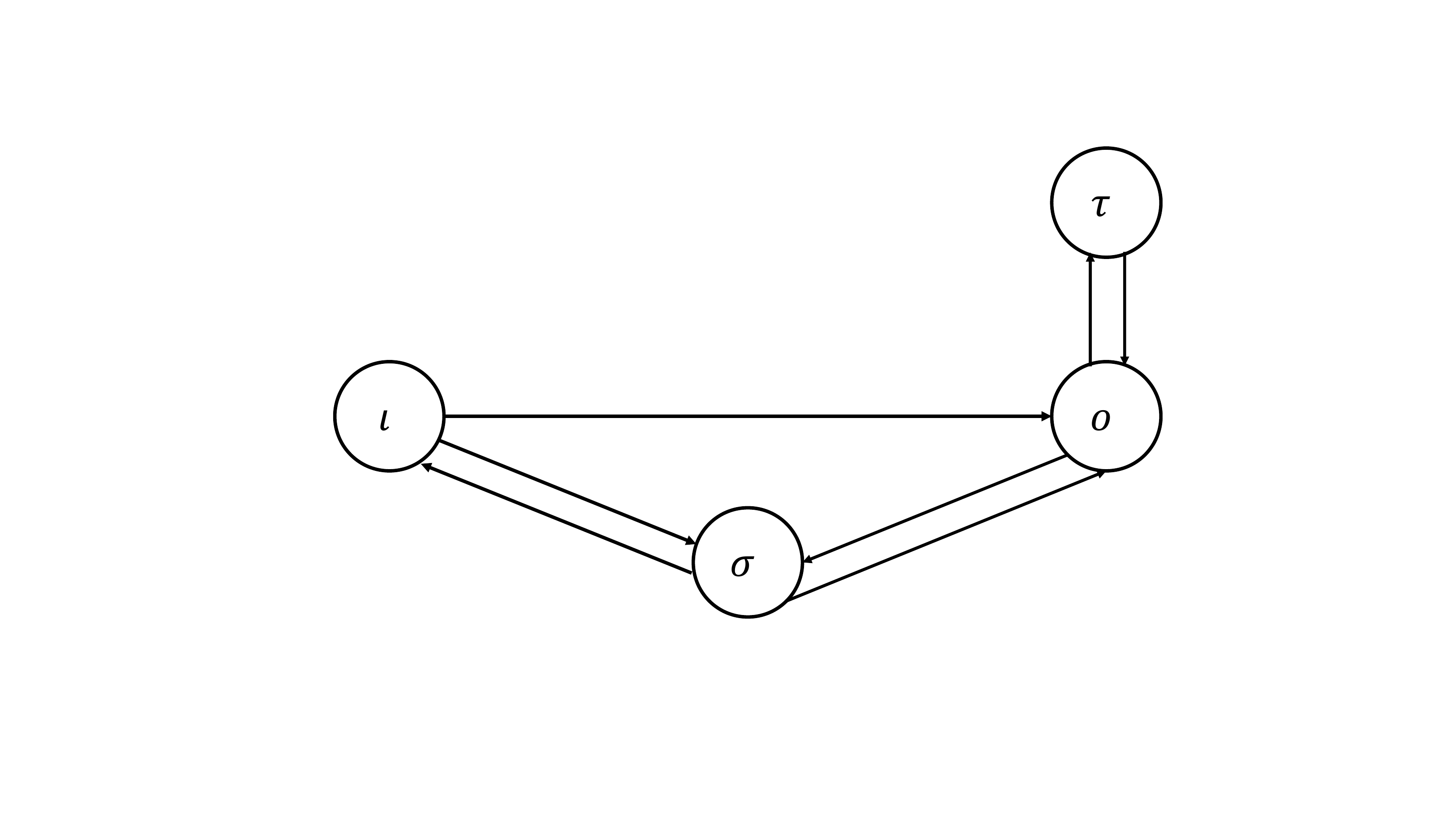} \\
(A) & (B)
\end{tabular}
\caption{\label{four_nodes_example_feedback}
Two abstract four-node input-output networks considered in this paper.
(A) Intracellular copper regulation (core equivalent to network number 20 in \cite{huang22}).
(B) Self immune recognition (core equivalent to network number 18 in \cite{huang22}).
}
\end{figure}

\subsection{Dynamical Formalism for Homeostasis}

Golubitsky and Stewart proposed a mathematical method for the study of homeostasis based on dynamical systems theory~\cite{gs17,gs18} (see the review \cite{gsahy20}). 
In this framework, one consider a system of differential equations 
\begin{equation} \label{general_dynamics}
  \dot{X} = F(X, \mathcal{I})
\end{equation}
where $X = (x_{1}, \cdots, x_{k}) \in \mathbb{R}^{k}$ and parameter $\mathcal{I}\in\mathbb{R}$ represents the external input to the system.

Suppose that $(X^*, \mathcal{I}^*)$ is a linearly stable equilibrium of \eqref{general_dynamics}.  
By the implicit function theorem, there is a function $\tilde{X}(\mathcal{I})$ defined in a neighborhood of $\mathcal{I}^*$ such that $\tilde{X}(\mathcal{I}^*) = X^*$ and $F(\tilde{X}(\mathcal{I}), \mathcal{I}) \equiv 0$.  
The simplest case is when there is a variable, let's say $x_{k}$, whose output is of interest when $\mathcal{I}$ varies.  
Define the associated \emph{input-output function} as $z(\mathcal{I})=\tilde{x}_k(\mathcal{I})$. 
The input-output function allows one to formulate several definitions that capture the notion of homeostasis (see~\cite{ma09,ang13,tang16,gs17,gs18}).

Let $z(\mathcal{I})$ be the input-output function associated to a system of differential equations \eqref{general_dynamics} and the family of equilibria $\tilde{X}(\mathcal{I})$.
We say that the corresponding system \eqref{general_dynamics} exhibits
\begin{enumerate}[(a)]
\item \emph{Perfect Homeostasis (Adaptation)} on the interval $(\mathcal{I}_{1}, \mathcal{I}_{2})$ if
\begin{equation} \label{definition_perfect_adaptation}
  \frac{d z}{d \mathcal{I}} (\mathcal{I}) = 0 
  \qquad\text{for all} \; \mathcal{I} \in (\mathcal{I}_{1}, \mathcal{I}_{2})
\end{equation}
That is, $z$ is constant on $(\mathcal{I}_{1}, \mathcal{I}_{2})$.
\item \emph{Near-perfect Homeostasis (Adaptation)} relative to a \emph{set point} $\mathcal{I}_{\mathrm{sp}}$ on the interval $(\mathcal{I}_{1}, \mathcal{I}_{2})$ if, for a fixed $\delta$,
\begin{equation} \label{definition_near_perfect_adaptation}
  | z(\mathcal{I}) - z(\mathcal{I}_{\mathrm{sp}}) | \leqslant \delta
  \qquad\text{for all} \; \mathcal{I} \in (\mathcal{I}_{1}, \mathcal{I}_{2})
\end{equation}
That is, $z$ stays within $z(\mathcal{I}_{\mathrm{sp}})\pm\delta$ over $(\mathcal{I}_{1}, \mathcal{I}_{2})$.
\item \emph{Infinitesimal Homeostasis} at the point $\mathcal{I}_{\mathrm{c}}$ on the interval $(\mathcal{I}_{1}, \mathcal{I}_{2})$ if
\begin{equation} \label{definition_homeostasis}
  \frac{d z}{d \mathcal{I}} (\mathcal{I}_{\mathrm{c}}) = 0
\end{equation}
That is, $\mathcal{I}_{\mathrm{c}}$ is a \emph{critical point} of
$z$.
\end{enumerate}

It is clear that perfect homeostasis implies near-perfect homeostasis, but the converse does not hold.  
Inspired by Reed \etal \cite{nijhout14,best09}, Golubitsky and Stewart \cite{gs17,gs18} introduced the notion of infinitesimal homeostasis that is intermediate between perfect and near-perfect homeostasis.
It is obvious that perfect homeostasis implies infinitesimal homeostasis. 
On the other hand, it follows from Taylor's theorem that infinitesimal homeostasis implies near-perfect homeostasis in a neighborhood of $\mathcal{I}_0$. 
It is easy to see that the converse to both implications is not generally valid (see \cite{reed17}).
Moreover, the notion of infinitesimal homeostasis allows the tools from singularity theory to bear on the study of homeostasis.

When combined with coupled systems theory~\cite{gs06} the formalism of \cite{gs17,gs18,gsahy20} becomes very effective in the analysis of model equations.

An \emph{input-output network} is a network $\mathcal{G}$ with a distinguished \emph{input node} $\iota$, associated to the input parameter $\mathcal{I}$, one distinguished \emph{output node} $o$, and $N$ \emph{regulatory nodes} $\rho=\{\rho_1,\ldots,\rho_N\}$.
The associated network systems of differential equations have the form
\begin{equation} \label{admissible_systems_ODE_multiple_input_nodes}
\begin{aligned}
\dot{x}_{\iota} & = f_{\iota}(x_{\iota}, x_{\rho}, x_{o}, \mathcal{I}) \\
\dot{x}_{\rho} & = f_{\rho}(x_{\iota}, x_{\rho}, x_{o})\\
\dot{x}_{o} & = f_{o}(x_{\iota}, x_{\rho}, x_{o})
\end{aligned}
\end{equation}
where $\mathcal{I}\in\mathbb{R}$ is an \emph{external input parameter} and $X=(x_{\iota},x_{\rho},x_o)\in\mathbb{R}\times\mathbb{R}^N\times\mathbb{R}$ is the vector of state variables associated to the network nodes. We write a vector field associated with the system \eqref{admissible_systems_ODE_multiple_input_nodes} as
\[
F(X,\mathcal{I})=(f_{\iota}(X,\mathcal{I}),f_\rho(X),f_o(X))
\]
and call it an \emph{admissible vector filed} for the network $\mathcal{G}$.

Let $f_{j,x_\ell}$ denote the partial derivative of the $j^{th}$ node function $f_j$ with respect to the $\ell^{th}$ node variable $x_\ell$.  We make the following assumptions about the vector field $F$ throughout:
\begin{enumerate}[(a)]
\item The vector field $F$ is smooth and has an asymptotically stable equilibrium at $(X^*,\mathcal{I}^*)$.
Therefore, by the implicit function theorem, there is a function $\tilde{X}(\mathcal{I})$ defined in a neighborhood of $\mathcal{I}^*$ such that $\tilde{X}(\mathcal{I}^*) = X^*$ and $F(\tilde{X}(\mathcal{I}), \mathcal{I}) \equiv 0$. 
\item The partial derivative $f_{j,x_\ell}$ can be non-zero only if the network $\mathcal{G}$ has an arrow $\ell\to j$, otherwise $f_{j,x_\ell} \equiv 0$.
\item Only the input node coordinate function $f_{\iota}$ depends on the external input parameter $\mathcal{I}$ and the partial derivative of $f_{\iota,\mathcal{I}}$ generically satisfies
\begin{equation} \label{e:f_iota_I}
 f_{\iota,\mathcal{I}} \neq 0.
\end{equation}
\end{enumerate}
The mapping $\mathcal{I} \mapsto x_o(\mathcal{I})$ is called the \emph{input-output function} of the input-output network $\mathcal{G}$ (associated to the family of equilibria $\tilde{X}(\mathcal{I})$).

As noted previously \cite{gs17,gsahy20,reed17,wang20}, a straightforward application of Cramer's rule gives a simple formula for determining infinitesimal homeostasis points.
Let $J$ be the $(N+2)\times (N+2)$ Jacobian matrix of an admissible vector field $F=(f_{\iota},f_{\sigma},f_{o})$, that is,
\begin{equation} \label{jacobian}
J = \begin{pmatrix}
  f_{\iota, x_\iota}   &  f_{\iota, x_\rho} & f_{\iota, x_o} \\
  f_{\rho, x_\iota}   &  f_{\rho, x_\rho} & f_{\rho, x_o} \\
  f_{o, x_\iota} &  f_{o, x_\rho} & f_{o, x_o} 
\end{pmatrix}
\end{equation}
The $(N+1)\times (N+1)$ matrix $H$ obtained from $J$ by dropping the last column and the first row is called \emph{homeostasis matrix} of $\mathcal{G}$:
\begin{equation} \label{homeostasis_matrix_definition}
H = 
\begin{pmatrix}
f_{\rho, x_\iota}&  f_{\rho, x_\rho} \\
f_{o, x_\iota} &  f_{o, x_\rho}
\end{pmatrix}
\end{equation}
In both eqs. \eqref{jacobian} and \eqref{homeostasis_matrix_definition} partial derivatives $f_{\ell,x_j}$ are evaluated at $\big(\tilde{X}(\mathcal{I}),\mathcal{I}\big)$.

\begin{lemma} \label{cramer_rule}
The input-output function $x_o(\mathcal{I})$ of an input-output network $\mathcal{G}$ satisfies
\begin{equation} \label{xo'}
x_o'(\mathcal{I}) = -f_{\iota,\mathcal{I}} \, \frac{\det(H)}{\det(J)}
\end{equation}
Here, $x_o'$ is the derivative of $x_o$ with respect to $\mathcal{I}$ and $\det(J)$, $\det(H)$ are evaluated at $\big(\tilde{X}(\mathcal{I}),\mathcal{I}\big)$.
Hence, $\mathcal{I}_0$ is a point of infinitesimal homeostasis if and only if
\begin{equation} \label{xo'_reduced}
\det(H) = 0
\end{equation}
at the equilibrium $\big(\tilde{X}(\mathcal{I}_0),\mathcal{I}_0\big)$.
\end{lemma}
\begin{proof}
See~\cite{gsahy20, wang20}.
\end{proof}

\section{Infinitesimal Homeostasis at a Boundary Point}

In this section we extend the theory of~\cite{gs17,gs18,wang20,gsahy20} to the case where the input-output function satisfies the near-perfect homeostasis condition on an open interval and the infinitesimal homeostasis occurs at a boundary point.

\subsection{Asymptotic Infinitesimal Homeostasis}

Consider a network $\mathcal{G}$ such that the associated input-output function $z$ is defined on a semi-infinite interval $D=(\mathcal{I}_{0}, +\infty)$.

\begin{theorem}
Let $z:D \to \mathbb{R}$ be a smooth function, with $D=(\mathcal{I}_{0}, +\infty)$.
Suppose that $z$ satisfies the near-perfect homeostasis condition on $D$: for all $\mathcal{I} \in D, z(\mathcal{I}) \in (z(\mathcal{I}_{\mathrm{sp}}) - \delta, z(\mathcal{I}_{\mathrm{sp}}) + \delta)$, for some $\mathcal{I}_{\mathrm{sp}} \in D$ and fixed $\delta > 0$. Then, at least one of the following statements is true:
\begin{enumerate}[(i)]
\item There exists $\mathcal{I}_{c} \in D$ such that $z'(\mathcal{I}_{c}) = 0$,
\item There exists an increasing sequence $(\mathcal{I}_{n})_{n \geqslant 1} \subset D$ satisfying
\[
\lim_{n \to \infty} \mathcal{I}_{n} = + \infty 
\qquad\text{and}\qquad  
\lim_{n \to \infty} z'(\mathcal{I}_{n}) = 0. 
\]
\end{enumerate}
In particular, if $z'$ is a monotonic function, then $\lim\limits_{\mathcal{I} \to +\infty} z'(\mathcal{I}) = 0$.
\end{theorem}
\begin{proof}
Suppose there exists $\mathcal{I}_{1}, \mathcal{I}_{2} \in D$ such that $z'(\mathcal{I}_{1}) \cdot z'(\mathcal{I}_{2}) \leqslant 0$. 
If $z'(\mathcal{I}_{1}) \cdot z'(\mathcal{I}_{2}) = 0$, then $z'(\mathcal{I}_{1}) = 0$ or $z'(\mathcal{I}_{2}) = 0$, and thus $(i)$ is true. On the other hand, if $z'(\mathcal{I}_{1}) \cdot z'(\mathcal{I}_{2}) < 0$, then $z'(\mathcal{I}_{1}) > 0$ and $z'(\mathcal{I}_{2}) < 0$ or $z'(\mathcal{I}_{1}) < 0$ and $z'(\mathcal{I}_{2}) > 0$. 
In both cases, by the mean value theorem, there exists $\mathcal{I}^{\star} \in (\mathcal{I}_{1}, \mathcal{I}_{2})$ such that $z'(\mathcal{I}^{\star}) = 0$, and thus $(i)$ is true.

Now suppose that for all $\mathcal{I}_{1}, \mathcal{I}_{2} \in D$, $z'(\mathcal{I}_{1}) \cdot z'(\mathcal{I}_{2}) > 0$. 
This means that $z'(\mathcal{I})$ is either positive or negative over $D$. 
Let us consider the case where $z'(\mathcal{I})$ is positive over $D$ (the other case is analogous). 
Since $z'(D)$ is bounded $\inf z'(D) \geq 0$.
Consider the dyadic sequence $\mathcal{J}_{n}=\sum_{m = 0}^{n} 2^{m}$, for $n \geqslant 0$, and define a family of consecutive disjoint intervals $(J_n)_{n\geqslant 1}$ contained in $D$, of length $2^n$, by
$J_n = \left(\mathcal{I}_{0}+\mathcal{J}_{n-1}, \mathcal{I}_{0} + \mathcal{J}_{n} \right)$.
Hence, one can write
\[
2^n \inf z'(J_n) =
\int_{\mathcal{I}_{0} + \mathcal{J}_{n-1}}^{\mathcal{I}_{0} + \mathcal{J}_{n}} \inf z'(J_n) \,\operatorname{d}\mathcal{I}  
\leqslant 
\int_{\mathcal{I}_{0} + \mathcal{J}_{n-1}}^{\mathcal{I}_{0} + \mathcal{J}_{n-1}} z'(\mathcal{I}) \,\operatorname{d}\mathcal{I}  
\leqslant 2 \delta
\]
and so
\[
 \inf z'(J_n) \leqslant \frac{\delta}{2^{n-1}}
\]
Therefore, there exists $\mathcal{I}_{n} \in J_{n}$, for $n \geqslant 1$, such that
\begin{equation} \label{limit_assimptotic_homeostasis}
 0 < z'(\mathcal{I}_{n}) \leqslant 
 \inf z'(J_n) + \frac{\delta}{2^{n-1}} \leqslant
 \frac{2\delta}{2^{n-1}}
\end{equation}
It is clear that $(\mathcal{I}_{n})_{n \geqslant 1}$ is an increasing sequence with $\lim\limits_{n \to \infty} \mathcal{I}_{n} = \infty$. 
By \eqref{limit_assimptotic_homeostasis}, we conclude that
\[
 \lim_{n \to \infty} z'(\mathcal{I}_{n}) = 0 
\]
and therefore $(ii)$ is true.
Finally, it is obvious that, if $z'$ is a monotonic function, then $\lim\limits_{\mathcal{I} \to \infty} z'(\mathcal{I}) = 0$.
\end{proof}

\begin{definition} \rm
An input-output function $z:D \to \mathbb{R}$, with $D=(\mathcal{I}_{0}, +\infty)$, exhibits \emph{asymptotic infinitesimal homeostasis} if it exhibits near-perfect homeostasis on $D$ and
\[
 \lim_{\mathcal{I} \to \infty} z'(\mathcal{I}) = 0.
\]
\end{definition}

\begin{corollary}
If an input-output function $z:D \to \mathbb{R}$, with $D=(\mathcal{I}_{0}, +\infty)$, exhibits near-perfect homeostasis and is monotonic then it exhibits asymptotic infinitesimal homeostasis.
\end{corollary}

\subsection{Core Networks and Asymptotic Infinitesimal Homeostasis}

Golubitsky et al.~\cite{wang20} have shown that in order to analyse if an input-output network exhibits infinitesimal homeostasis, it is enough to study an associated \emph{core network}, i.e., a network in which every node is downstream from the input node $\iota$ and upstream from the output node $o$. We will show that this theorem extends to the case of asymptotic infinitesimal homeostasis.

Let $\mathcal{G}$ be an input-output network with input node $\iota$, output node $o$ and regulatory nodes $\rho$. 
Partition the nodes of $\mathcal{G}$ three types:
\begin{itemize}
\item those nodes $\sigma$ that are both upstream from $o$ and downstream from $\iota$,
\item those nodes $d$ that are not downstream from $\iota$
\item those nodes $u$ which are downstream from $\iota$, but not upstream from $o$
\end{itemize}
Figure~\ref{core_network_abstract} exhibits this partition of regulatory nodes of $\mathcal{G}$.

\begin{figure}[!ht]
\centering
\includegraphics[scale = 0.3,trim=5cm 4cm 5cm 3cm,clip=true]%
{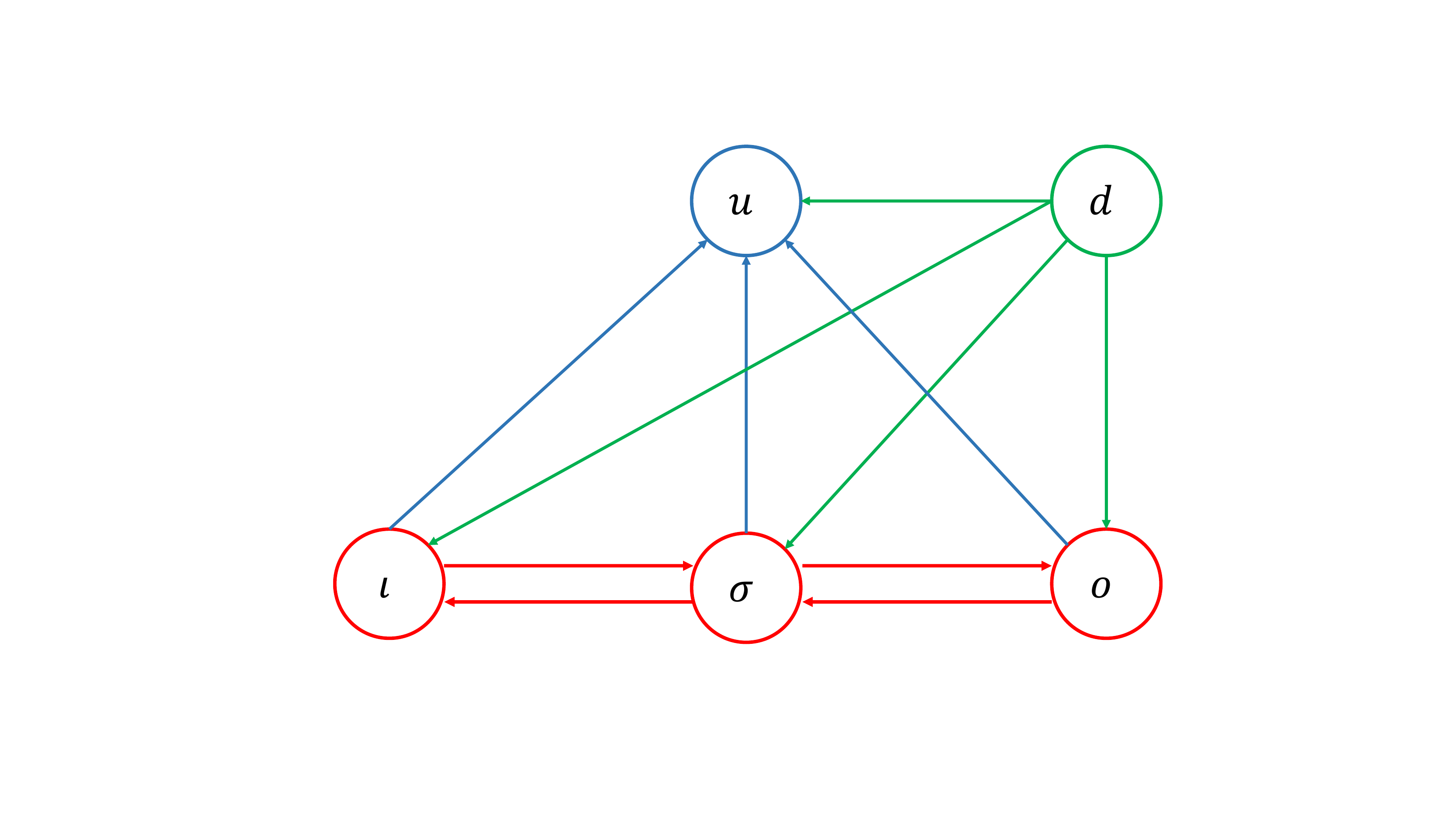}
\caption{\label{core_network_abstract}
Partition of nodes of $\mathcal{G}$.
Subnetwork in red is the core network $\mathcal{G}_{c}$.}
\end{figure}

The generic system of ODEs associated to the original network $\mathcal{G}$ is given by
\begin{equation} \label{admissible_system_original_network_abstract}
\begin{aligned}
    & \dot{x}_{\iota} = f_{\iota}(x_{\iota}, x_{\sigma}, x_{d}, x_{u}, x_{o}, \mathcal{I}) \\
    & \dot{x}_{\sigma} = f_{\sigma}(x_{\iota}, x_{\sigma}, x_{d}, x_{u}, x_{o}) \\
    & \dot{x}_{d} = f_{d}(x_{\iota}, x_{\sigma}, x_{d}, x_{u}, x_{o}) \\
    & \dot{x}_{u} = f_{u}(x_{\iota}, x_{\sigma}, x_{d}, x_{u}, x_{o}) \\
    & \dot{x}_{o} = f_{o}(x_{\iota}, x_{\sigma}, x_{d}, x_{u}, x_{o})
\end{aligned}
\end{equation}
The reduced systems of ODEs associated to the core network $\mathcal{G}_c$ obtained from \eqref{admissible_system_original_network_abstract} is given by
\begin{equation} \label{admissible_system_core_network_abstract}
\begin{aligned}
    & \dot{x}_{\iota} = f_{\iota}(x_{\iota}, x_{\sigma}, x_{d}, x_{o}, \mathcal{I}) \\
    & \dot{x}_{\sigma} = f_{\sigma}(x_{\iota}, x_{\sigma}, x_{d}, x_{o}) \\
    & \dot{x}_{d} = f_{d}(x_{d}) \\
    & \dot{x}_{u} = f_{u}(x_{\iota}, x_{\sigma}, x_{d}, x_{u}, x_{o}) \\
    & \dot{x}_{o} = f_{o}(x_{\iota}, x_{\sigma}, x_{d}, x_{o})
\end{aligned}
\end{equation}

\begin{theorem} \label{thm-core-networks}
Let $x_{o}(\mathcal{I})$ be the input-output function of the admissible system \eqref{admissible_system_original_network_abstract} and let $x_{o}^{c}(\mathcal{I})$ be the input-output function of the associated core admissible system \eqref{admissible_system_core_network_abstract}. Consider that both functions are defined in the semi-infinite interval $D=(\mathcal{I}_{0}, +\infty)$. Then, the input-output function $x_{o}^{c}(\mathcal{I})$ associated to the core subnetwork $\mathcal{G}_c$ exhibits asymptotic infinitesimal homeostasis if and only if the input-output function $x_{o}(\mathcal{I})$ associated to the original network $\mathcal{G}$ exhibits asymptotic infinitesimal homeostasis.
\end{theorem}
\begin{proof}
The Jacobian $J$ the original network is
\begin{equation}
J = \begin{pmatrix}
    f_{\iota, x_{\iota}} & f_{\iota, x_{\sigma}} & f_{\iota, x_{d}} & 0 & f_{\iota, x_{o}} \\
    f_{\sigma, x_{\iota}} & f_{\sigma, x_{\sigma}} & f_{\sigma, x_{d}} & 0 & f_{\sigma, x_{o}} \\
    0 & 0 & f_{d, x_{d}} & 0 & 0 \\
    f_{u, x_{\iota}} & f_{u, x_{\sigma}} & f_{u, x_{d}} & f_{u, x_{u}} & f_{u, x_{o}} \\
    f_{o, x_{\iota}} & f_{o, x_{\sigma}} & f_{o, x_{d}} & 0 & f_{o, x_{o}}
\end{pmatrix}
\end{equation}
and the corresponding homeostasis matrix $H$ is
\begin{equation}
H = \begin{pmatrix}
    f_{\sigma, x_{\iota}} & f_{\sigma, x_{\sigma}} & f_{\sigma, x_{d}} & 0 \\
    0 & 0 & f_{d, x_{d}} & 0 \\
    f_{u, x_{\iota}} & f_{u, x_{\sigma}} & f_{u, x_{d}} & f_{u, x_{u}} \\
    f_{o, x_{\iota}} & f_{o, x_{\sigma}} & f_{o, x_{d}} & 0
\end{pmatrix}
\end{equation}
On the other hand, the Jacobian $J^{c}$ and the homeostasis matrix $H^{c}$ of the core network are, respectively:
\begin{equation}
J^{c} = \begin{pmatrix}
    f_{\iota, x_{\iota}} & f_{\iota, x_{\sigma}} & f_{\iota, x_{o}} \\
    f_{\sigma, x_{\iota}} & f_{\sigma, x_{\sigma}} & f_{\sigma, x_{o}} \\
    f_{o, x_{\iota}} & f_{o, x_{\sigma}} & f_{o, x_{o}}
\end{pmatrix}
\qquad\text{and}\qquad
H^{c} = \begin{pmatrix}
    f_{\sigma, x_{\iota}} & f_{\sigma, x_{\sigma}} \\
    f_{o, x_{\iota}} & f_{o, x_{\sigma}}
\end{pmatrix}
\end{equation}
Then we can compute
\begin{equation}
\begin{split}
    \det H & = (-1)^{k_{H}}\det(f_{d, x_{d}})\det(f_{u, x_{u}})\det H^{c} \\
    \det J & = (-1)^{k_{J}}\det(f_{d, x_{d}})\det(f_{u, x_{u}})\det J^{c}
\end{split}
\end{equation}
Now, for all $\mathcal{I} \in D$, $J$ and $J^{c}$ must have eigenvalues with negative real part. As the eigenvalues of $f_{d, x_{d}}$ and of $f_{u, x_{u}}$ are also eigenvalues of $J$, we conclude that
\begin{equation}
    \det(f_{d, x_{d}}) \cdot \det(f_{u, x_{u}}) \neq 0
\end{equation}
Therefore
\begin{equation}
    \lim_{\mathcal{I} \to \infty} \frac{\det H}{\det J} =  \lim_{\mathcal{I} \to \infty} \frac{(-1)^{k_{H}}\det(f_{d, x_{d}})\det(f_{u, x_{u}})\det H^{c}}{(-1)^{k_{J}}\det(f_{d, x_{d}})\det(f_{u, x_{u}})\det J^{c}} = (-1)^{k} \lim_{\mathcal{I} \to \infty} \frac{\det H^{c}}{\det J^{c}}
\end{equation}
which concludes the proof.
\end{proof}

\begin{remark} \rm
The results of this section were obtained by considering an input-output function $z$ defined on a semi-infinite interval $D=(\mathcal{I}_{0}, +\infty)$. However, it is easy to see that they can be extended to the case where $z$ is defined on any finite open interval $D=(\mathcal{I}_{0}, \mathcal{I}_{\star})$, where $\mathcal{I}_{\star}$ is the point of asymptotic infinitesimal homeostasis, namely,
\[
 \lim_{\mathcal{I} \to \mathcal{I}_{\star}} z'(\mathcal{I}) = 0.
\]
In any case the point of infinitesimal homeostasis is on the boundary of definition of the input-output function. \END
\end{remark}

\section{Self Immune Recognition}

\subsection{Brief Review of Immune Recognition}

The immune system has a paramount role in mammalian physiology: it must combat any strange body and infection, and, at the same time, it must discriminate between which elements belong to the organism and which not in order to avoid autoimmunity, something know in the literature as self and non-self recognition. Although specificity of receptors expressed by immune cells is a major mechanism that explains the capacity of discrimination between self and non-self components, conventional T lymphocytes in tissues may still be erroneously activated leading to autoimmunity and cell injury \cite{abbas14}.

Let's consider here the three main immune cells that are present in tissues: antigen-presenting cells (APCs), responsible for initiating the immune response, conventional T lymphocytes (Tconv), which are the main cells responsible for a specific response against non-self pathogens, and regulatory T lymphocytes(Treg), which suppress Tconv activity \cite{abbas14, sok18, khailaie13}. Usually, the immune response starts when APCs take digested antigens and couple them to MHC molecules expressed in APCs surface \cite{khailaie13}. This enables the recognition of the antigen by Tconv. When activates, Tconv cells synthesize interleukin-2 (IL2), which stimulates both Tconv and Treg cells. On the other hand, Treg interacts to Tconv, particularly with autoreactive Tconv, supressing their activity \cite{abbas14, khailaie13}.

The importance of Treg cells may be exemplified by the fact that patients with pathogenic variants in \textit{FOXP3} gene leading to Treg cells dysfunction develop an autoimmune syndrome called IPEX (Immune dysregulation, polyendocrinopathy, enteropathy, X-linked syndrome) \cite{bacchetta18}.

\subsection{Mathematical Model}

In order to evaluate how the concept of infinitesimal homeostasis could be applied in the context of autoimmune activation, we shall adapt a model previously published by Khailaie \etal \cite{khailaie13}, that considers a situation where the only existing antigen are self. This version of the model is given by the interplay between four components: APCs, Tconv, IL2 and Treg. Representing the dimensionless concentrations of APCs, Tconv, IL2 and Treg by, respectively, $x_{\tau}, x_{o}, x_{\sigma}$ and $x_{\iota}$, with $\mathcal{I}$ the input parameter, the dynamics is described by the systems of ODEs
\begin{equation} \label{simplified_immune_model}
\begin{aligned}
& \dot{x}_\iota = a x_\iota x_\sigma - b x_\iota + \mathcal{I}  \\
& \dot{x}_\sigma = c x_o - d x_\sigma (x_o + x_{\iota}) - e x_\sigma \\
& \dot{x}_\tau = - b x_\tau + f \frac{x_\tau}{x_\tau + g} + h x_\tau x_o + j \\
& \dot{x}_o = a x_\sigma x_o - b x_o - l x_\iota x_o + h x_\tau x_o + j
\end{aligned}
\end{equation}
where $a, b, c, d, e, f, g, h, j$ and $l$ are positive parameters.
Considering $\iota$, $\sigma$, $\tau$ and $o$ as the nodes of a network we obtain the network in Figure \ref{four_nodes_example_feedback}(B).

\subsection{Infinitesimal Homeostasis}

The model \eqref{simplified_immune_model} can have only two types of homeostasis: structural and null-degradation. The homeostasis matrix $H$ of the network is
\begin{equation}
  H = \begin{pmatrix}
  f_{\iota, \iota} & f_{\iota, \sigma} & 0 & 1 \\
  f_{\sigma, \iota} & f_{\sigma, \sigma} & 0 & 0 \\
  0 & 0 & f_{\tau, \tau} & 0 \\
  f_{o, \iota} & f_{o, \sigma} & f_{o, \tau}  & 0
  \end{pmatrix}
\end{equation}
Thus
\begin{equation}
    \det H = f_{\tau, \tau}(f_{o, \sigma}f_{\sigma, \iota} - f_{o, \iota}f_{\sigma, \sigma})
\end{equation}

Let us show that $\det H \neq 0$ for all $\mathcal{I} \in \mathbb{R}^{+}$ at any equilibrium point. In fact, considering the ODE in \eqref{simplified_immune_model}, in any equilibrium we must have $x_o \neq 0$ (we may conclude this looking to the equation that defines $\dot{x}_o$). Fixing always the initial state as $(0,0,0,0)$, then it is easy to verify that it is plausible to assume that $x_{\tau}, x_{o}, x_{\sigma}$ and $x_{\iota}$ must be non-negative at equilibrium. 
Now, observe that
\begin{equation}
    \begin{aligned}
    & f_{\sigma, \iota} = -ex_\sigma \\
    & f_{\sigma, \sigma} = -e(x_o + x_\iota) - f \\
    & f_{o, \iota} = - l x_o \\
    & f_{o, \sigma} = ax_{o}
    \end{aligned}
\end{equation}
Therefore
\begin{equation} \label{structural_homeostasis_immune}
\begin{aligned}
    & f_{o, \sigma}f_{\sigma, \iota} - f_{o, \iota}f_{\sigma, \sigma} = ax_{o}(-ex_\sigma) - l x_o \left[e(x_o + x_\iota) + f \right] \\
    & \Rightarrow f_{o, \sigma}f_{\sigma, \iota} - f_{o, \iota}f_{\sigma, \sigma} = - x_{o} [aex_\sigma + lex_\iota + lex_o + lf]
\end{aligned}
\end{equation}
As for all $\mathcal{I} > 0$ for which the system admits a linearly stable equilibrium, the equilibrium point $(\tilde{x}_\iota, \tilde{x}_\sigma, \tilde{x}_\tau, \tilde{x}_o) \in \mathbb{R}^{*}_{+}$, we conclude that the system does not present structural homeostasis.

We shall now prove that it does not present null degradation homeostasis neither. Suppose that there is an equilibrium $(\tilde{x}_\iota, \tilde{x}_\sigma, \tilde{x}_\tau, \tilde{x}_o)$ such that it satisfies $f_{\tau, \tau} = 0$
\begin{equation} \label{immune_model_null_degradation_part_i}
f_{\tau, \tau} = 0 \Rightarrow -b + \frac{fg}{(\tilde{x}_\tau + g)^{2}} + h \tilde{x}_o = 0 \Rightarrow \tilde{x}_o = \frac{b}{h} - \frac{fg}{h(\tilde{x}_\tau + g)^{2}}
\end{equation}
Applying \eqref{immune_model_null_degradation_part_i} to the fact that it must happen in an equilibrium point
\begin{equation}
\begin{split}
    \dot{x}_\tau = 0  & \Rightarrow - b \tilde{x}_\tau + f \frac{\tilde{x}_\tau}{\tilde{x}_\tau + g} + h \tilde{x}_\tau \tilde{x}_o + j = 0 \\
    & \Rightarrow - b \tilde{x}_\tau + f \frac{\tilde{x}_\tau}{\tilde{x}_\tau + g} + h \tilde{x}_\tau \left( \frac{b}{h} - \frac{fg}{h(\tilde{x}_\tau + g)^{2}} \right) + j = 0\\
    & \Rightarrow f \frac{\tilde{x}_\tau}{\tilde{x}_\tau + g} - \frac{fg\tilde{x}_\tau}{(\tilde{x}_\tau + g)^{2}} + j = 0 \\
    & \Rightarrow \frac{f\tilde{x}_\tau(\tilde{x}_\tau + g) - fg\tilde{x}_\tau + j(\tilde{x}_\tau + g)^{2}}{(\tilde{x}_\tau + g)^{2}} = 0 \\
    & \Rightarrow f \tilde{x}_\tau^{2} + j(\tilde{x}_\tau + g)^{2} = 0
\end{split}
\end{equation}
As the last equality cannot hold, the system does not exhibit null degradation homeostasis in node $\tau$.

We already know that the system does not exhibit infinitesimal homeostasis for $\mathcal{I} \in \mathbb{R}^{+}$. Let us study now what happens when $\mathcal{I} \rightarrow + \infty$. For this, we have to write the equilibrium points $(\tilde{x}_\iota, \tilde{x}_\sigma, \tilde{x}_\tau, \tilde{x}_o)$ as a function of $\mathcal{I}$.
First, taking the differential equation for $\dot{x}_\iota$ \eqref{simplified_immune_model}, we conclude that for $\mathcal{I} > 0$, $\tilde{x}_\iota \neq 0$. Consequently, we conclude that
\begin{equation} \label{interleukin_function_Treg}
   \dot{x}_\iota = 0 \Rightarrow a \tilde{x}_\iota \tilde{x}_\sigma - b \tilde{x}_\iota + \mathcal{I} = 0 \Rightarrow \tilde{x}_\sigma = \frac{b\tilde{x}_\iota - \mathcal{I}}{a \tilde{x}_\iota}
\end{equation}
Applying \eqref{interleukin_function_Treg} to the dynamics of $\dot{x}_\sigma$, we obtain
\begin{equation} \label{formula_x_o_immune}
\begin{split}
    \dot{x}_\sigma = 0 
    & \Rightarrow c x_o - d x_\sigma (x_o + x_{\iota}) - e x_\sigma = 0  
    \Rightarrow \tilde{x}_o (c - d\tilde{x}_\sigma) = \tilde{x}_\sigma (d\tilde{x}_\iota + e) \\
    & \Rightarrow \tilde{x}_o = \frac{ \tilde{x}_\sigma (d\tilde{x}_\iota + e)}{c - d\tilde{x}_\sigma} 
    \Rightarrow \tilde{x}_o = \frac{ \displaystyle \left(\frac{b\tilde{x}_\iota - \mathcal{I}}{a \tilde{x}_\iota}\right) (d\tilde{x}_\iota + e)}{\displaystyle c - d\left(\frac{b\tilde{x}_\iota - \mathcal{I}}{a \tilde{x}_\iota}\right)} \\
    & \Rightarrow \tilde{x}_{o} = \frac{(d\tilde{x}_\iota+e)(b\tilde{x}_\iota- \mathcal{I})}{ac\tilde{x}_\iota -d(b\tilde{x}_\iota - \mathcal{I})}
\end{split}
\end{equation}
Considering now the dynamics of $\dot{x}_o$:
\begin{equation}
    \dot{x}_o = 0 \Rightarrow a \tilde{x}_\sigma \tilde{x}_o - b \tilde{x}_o - l \tilde{x}_\iota \tilde{x}_o + h \tilde{x}_\tau \tilde{x}_o + j = 0
\end{equation}
As mentioned before, $j > 0 \Rightarrow \tilde{x}_o \neq 0$, and so
\begin{equation} \label{formula_xtau}
    a \tilde{x}_\sigma \tilde{x}_o - b \tilde{x}_o - l \tilde{x}_\iota \tilde{x}_o + h \tilde{x}_\tau \tilde{x}_o + j = 0 \Rightarrow \tilde{x}_\tau = \frac{b - a \tilde{x}_\sigma + l\tilde{x}_\iota}{h} - \frac{j}{h\tilde{x}_o} 
\end{equation}
Notice that, by \eqref{interleukin_function_Treg}, we have
\begin{equation} \label{simplification-b-axsigma}
    b - a\tilde{x}_\sigma = b - a\frac{b\tilde{x}_\iota - \mathcal{I}}{a \tilde{x}_\iota} = \frac{\mathcal{I}}{\tilde{x}_\iota}
\end{equation}
And therefore \eqref{formula_xtau} is reduced to
\begin{equation} \label{formula_xtau_simplified}
    \tilde{x}_\tau = \frac{\mathcal{I}}{h\tilde{x}_\iota} + \frac{l\tilde{x}_\iota}{h} - \frac{j}{h\tilde{x}_o}
\end{equation}

Now, let's analyse the dynamics of $\dot{x}_\tau$, remembering that $j > 0 \Rightarrow \tilde{x}_\tau \neq 0$:
\begin{equation} \label{computation_x_i_infinity_part_i}
\begin{split}
    \dot{x}_\tau = 0 
    & \Rightarrow - b \tilde{x}_\tau + f \frac{\tilde{x}_\tau}{\tilde{x}_\tau + g} + h \tilde{x}_\tau \tilde{x}_o + j = 0 \\
    &\Rightarrow \tilde{x}_\tau (-b + h\tilde{x}_{o}) + f \frac{\tilde{x}_\tau}{\tilde{x}_\tau + g} + j = 0 \\
    & \Rightarrow b - h\tilde{x}_{o} = \frac{f}{\tilde{x}_\tau + g} + \frac{j}{\tilde{x}_\tau} \\
    & \Rightarrow b - h\tilde{x}_{o} = \frac{(f + j)\tilde{x}_\tau + gj}{\tilde{x}_\tau(\tilde{x}_\tau + g)} \\
    & \Rightarrow \tilde{x}_{o} = \frac{b}{h} - \frac{(f + j)\tilde{x}_\tau + gj}{h\tilde{x}_\tau(\tilde{x}_\tau + g)}
\end{split}
\end{equation}
In order to simplify the computations, let's suppose $bg = f + j$. In that case, we obtain, from \eqref{computation_x_i_infinity_part_i}:
\begin{equation}\label{computation_x_i_infinity_part_ii}
\begin{aligned}
    \tilde{x}_{o} = \frac{b\tilde{x}_{\tau}^{2} - jg}{h\tilde{x}_{\tau}^{2} + gh\tilde{x}_\tau} \Rightarrow \tilde{x}_{\tau}^{2}(h\tilde{x}_{o} - b) + gh\tilde{x}_\tau \tilde{x}_o + gj = 0
\end{aligned}
\end{equation}
Now, applying \eqref{formula_xtau_simplified} to \eqref{computation_x_i_infinity_part_ii}, we get:
\begin{equation} \label{computation_x_i_infinity_part_iii}
\begin{aligned}
   & (h\tilde{x}_{o} - b) \left( \frac{\mathcal{I}}{h\tilde{x}_\iota} + \frac{l\tilde{x}_\iota}{h} - \frac{j}{h\tilde{x}_o} \right)^{2} + gh\tilde{x}_{o}\left( \frac{\mathcal{I}}{h\tilde{x}_\iota} + \frac{l\tilde{x}_\iota}{h} - \frac{j}{h\tilde{x}_o} \right) + gj = 0 \\
   & (h\tilde{x}_{o} - b) \left( \frac{\mathcal{I}\tilde{x}_{o} + l\tilde{x}_{\iota}^{2}\tilde{x}_{o} - j\tilde{x}_{\iota}}{h\tilde{x}_{\iota}\tilde{x}_{o}}\right)^{2} + g\left( \frac{\mathcal{I}\tilde{x}_{o} + l\tilde{x}_{\iota}^{2}\tilde{x}_{o}}{\tilde{x}_{\iota}}\right) = 0 \\
   & (h\tilde{x}_{o} - b)(\mathcal{I}\tilde{x}_{o} + l\tilde{x}_{\iota}^{2}\tilde{x}_{o} - j\tilde{x}_{\iota})^{2} + gh^{2}\tilde{x}_{o}^{3}\tilde{x}_{\iota}(\mathcal{I} + l\tilde{x}_{\iota}^{2}) = 0
\end{aligned}
\end{equation}
Writing $\tilde{x}_{o}$ in function of $\mathcal{I}$ and $\tilde{x}_{\iota}$ according to \eqref{formula_x_o_immune} in \eqref{computation_x_i_infinity_part_iii}, we obtain
\begin{equation} \label{computation_x_i_infinity_part_iv}
\setlength{\jot}{15pt}
\begin{split}
    & \left[ \frac{h(d\tilde{x}_\iota+e)(b\tilde{x}_\iota- \mathcal{I})}{ac\tilde{x}_\iota -d(b\tilde{x}_\iota - \mathcal{I})} - b\right] \left[ \frac{\mathcal{I}(d\tilde{x}_\iota+e)(b\tilde{x}_\iota- \mathcal{I})}{ac\tilde{x}_\iota -d(b\tilde{x}_\iota - \mathcal{I})} + \frac{l \tilde{x}_{\iota}^{2}(d\tilde{x}_\iota+e)(b\tilde{x}_\iota- \mathcal{I})}{ac\tilde{x}_\iota -d(b\tilde{x}_\iota - \mathcal{I})} -j\tilde{x}_{\iota}\right]^{2} \\
    & + gh^{2} \frac{(d\tilde{x}_\iota+e)^{3}(b\tilde{x}_\iota- \mathcal{I})^{3}}{[ac\tilde{x}_\iota -d(b\tilde{x}_\iota - \mathcal{I})]^{3}}\tilde{x}_{\iota}(\mathcal{I} + l \tilde{x}_{\iota}^{2}) = 0 \\
    & \left[ \frac{h(d\tilde{x}_\iota+e)(b\tilde{x}_\iota- \mathcal{I})- bac\tilde{x}_\iota + bd(b\tilde{x}_\iota - \mathcal{I})}{ac\tilde{x}_\iota -d(b\tilde{x}_\iota - \mathcal{I})}\right] \cdot \frac{1}{[ac\tilde{x}_\iota -d(b\tilde{x}_\iota - \mathcal{I})]^{2}} \cdot [\mathcal{I}(d\tilde{x}_\iota+e)(b\tilde{x}_\iota- \mathcal{I}) \\
    & + l \tilde{x}_{\iota}^{2}(d\tilde{x}_\iota+e)(b\tilde{x}_\iota- \mathcal{I}) - acj \tilde{x}_{\iota}^{2} + dj\tilde{x}_{\iota}(b\tilde{x}_\iota- \mathcal{I})]^{2} + gh^{2} \frac{(d\tilde{x}_\iota+e)^{3}(b\tilde{x}_\iota- \mathcal{I})^{3}}{[ac\tilde{x}_\iota -d(b\tilde{x}_\iota - \mathcal{I})]^{3}}\tilde{x}_{\iota}(\mathcal{I} + l \tilde{x}_{\iota}^{2}) = 0 \\
    & [h(d\tilde{x}_\iota+e)(b\tilde{x}_\iota- \mathcal{I})- bac\tilde{x}_\iota + bd(b\tilde{x}_\iota - \mathcal{I})][\mathcal{I}(d\tilde{x}_\iota+e)(b\tilde{x}_\iota- \mathcal{I}) + l \tilde{x}_{\iota}^{2}(d\tilde{x}_\iota+e)(b\tilde{x}_\iota- \mathcal{I}) \\
    & - acj \tilde{x}_{\iota}^{2} + dj\tilde{x}_{\iota}(b\tilde{x}_\iota- \mathcal{I})]^{2} + gh^{2}(d\tilde{x}_\iota+e)^{3}(b\tilde{x}_\iota- \mathcal{I})^{3} \tilde{x}_{\iota}(\mathcal{I} + l \tilde{x}_{\iota}^{2}) = 0 \\
    & [h(bd\tilde{x}_{\iota}^{2} -d\mathcal{I}\tilde{x}_{\iota} + eb\tilde{x}_{\iota} - e\mathcal{I}) - bac\tilde{x}_\iota + bd(b\tilde{x}_\iota - \mathcal{I})][\mathcal{I}(d\tilde{x}_\iota+e)(b\tilde{x}_\iota- \mathcal{I}) + l \tilde{x}_{\iota}^{2}(d\tilde{x}_\iota+e)(b\tilde{x}_\iota- \mathcal{I}) \\
    & - acj \tilde{x}_{\iota}^{2} + dj\tilde{x}_{\iota}(b\tilde{x}_\iota- \mathcal{I})]^{2} + gh^{2}(d\tilde{x}_\iota+e)^{3}(b\tilde{x}_\iota- \mathcal{I})^{3} \tilde{x}_{\iota}(\mathcal{I} + l \tilde{x}_{\iota}^{2}) = 0 \\
    & [h(bd\tilde{x}_{\iota}^{2} + eb\tilde{x}_{\iota} - e\mathcal{I}) -\tilde{x}_{\iota} (dh\mathcal{I} + bac) + bd(b\tilde{x}_\iota - \mathcal{I})][\mathcal{I}(d\tilde{x}_\iota+e)(b\tilde{x}_\iota- \mathcal{I}) + l \tilde{x}_{\iota}^{2}(d\tilde{x}_\iota+e)(b\tilde{x}_\iota- \mathcal{I}) \\
    & - acj \tilde{x}_{\iota}^{2} + dj\tilde{x}_{\iota}(b\tilde{x}_\iota- \mathcal{I})]^{2} + gh^{2}(d\tilde{x}_\iota+e)^{3}(b\tilde{x}_\iota- \mathcal{I})^{3} \tilde{x}_{\iota}(\mathcal{I} + l \tilde{x}_{\iota}^{2}) = 0
\end{split}
\end{equation}
Thus $\displaystyle \lim_{\mathcal{I} \rightarrow +\infty} dh\mathcal{I} + bac = dh\mathcal{I}$. Therefore, when we take the limit $\mathcal{I} \rightarrow + \infty$, the polynomial equation described on \eqref{computation_x_i_infinity_part_iv} has the same solutions as
\begin{equation} \label{computation_x_i_infinity_part_v}
\setlength{\jot}{15pt}
\begin{split}
    & [h(bd\tilde{x}_{\iota}^{2} + eb\tilde{x}_{\iota} - e\mathcal{I}) -\tilde{x}_{\iota} (dh\mathcal{I}) + bd(b\tilde{x}_\iota - \mathcal{I})][\mathcal{I}(d\tilde{x}_\iota+e)(b\tilde{x}_\iota- \mathcal{I}) \\
    & + l\tilde{x}_{\iota}^{2}(d\tilde{x}_\iota+e)(b\tilde{x}_\iota- \mathcal{I}) - acj \tilde{x}_{\iota}^{2} + dj\tilde{x}_{\iota}(b\tilde{x}_\iota- \mathcal{I})]^{2} \\
    & + gh^{2}(d\tilde{x}_\iota+e)^{3}(b\tilde{x}_\iota- \mathcal{I})^{3} \tilde{x}_{\iota}(\mathcal{I} + l \tilde{x}_{\iota}^{2}) = 0 \\
    & [h(d\tilde{x}_\iota+e)(b\tilde{x}_\iota- \mathcal{I}) + bd(b\tilde{x}_\iota - \mathcal{I})][\mathcal{I}(d\tilde{x}_\iota+e)(b\tilde{x}_\iota- \mathcal{I}) + l \tilde{x}_{\iota}^{2}(d\tilde{x}_\iota+e)(b\tilde{x}_\iota- \mathcal{I}) \\
    & - acj \tilde{x}_{\iota}^{2} + dj\tilde{x}_{\iota}(b\tilde{x}_\iota- \mathcal{I})]^{2} + gh^{2}(d\tilde{x}_\iota+e)^{3}(b\tilde{x}_\iota- \mathcal{I})^{3} \tilde{x}_{\iota}(\mathcal{I} + l \tilde{x}_{\iota}^{2}) = 0 \\
    & (b\tilde{x}_\iota- \mathcal{I}) \{[h(d\tilde{x}_\iota+e) + bd][\mathcal{I}(d\tilde{x}_\iota+e)(b\tilde{x}_\iota- \mathcal{I}) + l \tilde{x}_{\iota}^{2}(d\tilde{x}_\iota+e)(b\tilde{x}_\iota- \mathcal{I}) - acj \tilde{x}_{\iota}^{2} \\
    & + dj\tilde{x}_{\iota}(b\tilde{x}_\iota- \mathcal{I})]^{2} 
    + gh^{2}(d\tilde{x}_\iota+e)^{3}(b\tilde{x}_\iota- \mathcal{I})^{2} \tilde{x}_{\iota}(\mathcal{I} + l \tilde{x}_{\iota}^{2})\} = 0
\end{split}
\end{equation}
Therefore, we got one of the roots of \eqref{computation_x_i_infinity_part_v}
\begin{equation} \label{value_x_iota_immune_infinity}
    \lim_{\mathcal{I} \rightarrow +\infty} b\tilde{x}_\iota- \mathcal{I} = 0 \Rightarrow \lim_{\mathcal{I} \rightarrow +\infty} \tilde{x}_\iota = +\infty
\end{equation}

Applying the result of \eqref{value_x_iota_immune_infinity} to \eqref{interleukin_function_Treg} and \eqref{formula_x_o_immune}, we obtain
\begin{equation} \label{computation_x_sigma_x_o_infinity_immune}
\setlength{\jot}{15pt}
\begin{aligned}
    & \lim_{\mathcal{I} \rightarrow +\infty} \tilde{x}_\sigma = \lim_{\mathcal{I} \rightarrow +\infty} \frac{b\tilde{x}_\iota - \mathcal{I}}{a \tilde{x}_\iota} \Rightarrow \lim_{\mathcal{I} \rightarrow +\infty} \tilde{x}_\sigma = 0 \\
    & \lim_{\mathcal{I} \rightarrow +\infty} \tilde{x}_o = \lim_{\mathcal{I} \rightarrow +\infty} \frac{(d\tilde{x}_\iota+e)(b\tilde{x}_\iota- \mathcal{I})}{ac\tilde{x}_\iota -d(b\tilde{x}_\iota - \mathcal{I})} = \lim_{\mathcal{I} \rightarrow +\infty} \frac{d\tilde{x}_\iota(b\tilde{x}_\iota- \mathcal{I})}{ac\tilde{x}_\iota} \Rightarrow \lim_{\mathcal{I} \rightarrow +\infty} \tilde{x}_o = 0
\end{aligned}
\end{equation}
Let us determine the value of $\tilde{x}_\tau$ at that equilibrium point. First, notice that $\displaystyle \lim_{\mathcal{I} \rightarrow +\infty} \neq \pm \infty$. In fact, suppose that $\displaystyle \lim_{\mathcal{I} \rightarrow +\infty} \tilde{x}_\tau = \pm \infty$ and consider the dynamics of $\dot{x}_\tau$
\begin{equation}
\begin{aligned}
    & \lim_{\mathcal{I} \rightarrow + \infty} \dot{x}_{\tau} = \lim_{\mathcal{I} \rightarrow + \infty} - b \tilde{x}_\tau + f \frac{\tilde{x}_\tau}{\tilde{x}_\tau + g} + h \tilde{x}_\tau \tilde{x}_o + j = \lim_{\mathcal{I} \rightarrow + \infty} \tilde{x}_\tau (- b + h \tilde{x}_o) + f \frac{\tilde{x}_\tau}{\tilde{x}_\tau + g} + j \\
    & = \lim_{\mathcal{I} \rightarrow + \infty} -b \tilde{x}_\tau + f + j = \lim_{\mathcal{I} \rightarrow + \infty} -b \tilde{x}_\tau = \pm \infty
\end{aligned}
\end{equation}
which is a contradiction since for all $\mathcal{I} \in \mathbb{R}^{*}_{+}$, at equilibrium, we have $\dot{x}_\tau = 0$. 

Therefore, $\tilde{x}_\tau$ must be limited when $\mathcal{I} \rightarrow +\infty \Rightarrow \displaystyle \lim_{\mathcal{I} \rightarrow + \infty} \tilde{x}_\tau \tilde{x}_o = 0$. Calling $\displaystyle \lim_{\mathcal{I} \rightarrow + \infty} \tilde{x}_\tau = \alpha$ and analysing again the dynamics of $\dot{x}_\tau$, we get
\begin{equation} \label{computation_limit_x_tau_immune_part_i}
\begin{split}
    - b \alpha + \frac{f \alpha}{\alpha + g} + j = 0 
    & \Rightarrow -b\alpha^{2} -bg\alpha +f\alpha +j\alpha +gj = 0 \\ 
    & \Rightarrow -b\alpha^{2} + \alpha (-bg + f + j) + gj = 0
\end{split}
\end{equation}
As we hypothesized before that $f + j = bg$, then \eqref{computation_limit_x_tau_immune_part_i} is reduced to
\begin{equation} \label{computation_limit_x_tau_immune_part_ii}
\begin{aligned}
    & -b\alpha^{2} + gj = 0 \Rightarrow \lim_{\mathcal{I} \rightarrow + \infty} \tilde{x}_\tau = \sqrt{\frac{gj}{b}}
\end{aligned}
\end{equation}

Let us now analyse if this equilibrium is linearly stable. For this purpose, we must study the behaviour of the Jacobian $J$ of the system when $\mathcal{I} \rightarrow +\infty$
\begin{equation} \label{behaviour_jacobian_infinity_immune_part_i}
\setlength{\jot}{15pt}
\begin{aligned}
& \lim_{\mathcal{I} \rightarrow +\infty} J = \lim_{\mathcal{I} \rightarrow +\infty} \begin{pmatrix}
  f_{\iota, \iota} & f_{\iota, \sigma} & 0 & 0 \\
  f_{\sigma, \iota} & f_{\sigma, \sigma} & 0 & f_{\sigma, o} \\
  0 & 0 & f_{\tau, \tau} & f_{\tau, o} \\
  f_{o, \iota} & f_{o, \sigma} & f_{o, \tau}  & f_{o, \tau}
  \end{pmatrix} \\
  & = \lim_{\mathcal{I} \rightarrow +\infty} \begin{pmatrix}
  a\tilde{x}_\sigma - b & a\tilde{x}_\iota & 0 & 0 \\
  -d\tilde{x}_\sigma & -d(\tilde{x}_o + \tilde{x}_\iota)-e & 0 & c - d\tilde{x}_\sigma \\
  0 & 0 & -b + \displaystyle \frac{fg}{(\tilde{x}_{\tau} + g)^{2}} + h\tilde{x}_o & h\tilde{x}_\tau \\
  -l\tilde{x}_o & a\tilde{x}_o & h\tilde{x}_o  & a\tilde{x}_\sigma - b -l\tilde{x}_\iota + h\tilde{x}_\tau
  \end{pmatrix}
\end{aligned}
\end{equation}
Applying the limits previously determined, we obtain
\begin{equation} \label{behaviour_jacobian_infinity_immune_part_ii}
\setlength{\jot}{15pt}
\begin{aligned}
& \lim_{\mathcal{I} \rightarrow +\infty} J = \lim_{\mathcal{I} \rightarrow +\infty} \begin{pmatrix}
  - b & a\tilde{x}_\iota & 0 & 0 \\
  0 & -d\tilde{x}_\iota & 0 & c \\
  0 & 0 & \displaystyle \frac{bf - b^{2}g - bj -2b\sqrt{bgj}}{gb + j + 2\sqrt{bgj}} & \displaystyle h\sqrt{\frac{gj}{b}} \\
  0 & 0 & 0  & -l\tilde{x}_\iota
  \end{pmatrix}
\end{aligned}
\end{equation}
Therefore, the eigenvalues of $\displaystyle \lim_{\mathcal{I} \rightarrow +\infty} J$ are $-b$, $\displaystyle \lim_{\mathcal{I} \rightarrow +\infty} -d\tilde{x}_\iota$, $\displaystyle \frac{bf - b^{2}g - bj -2b\sqrt{bgj}}{gb + j + 2\sqrt{bgj}}$ and $\displaystyle \lim_{\mathcal{I} \rightarrow +\infty} -l\tilde{x}_\iota$.
By hypothesis
\begin{equation}
f + j = bg \Rightarrow f < bg \Rightarrow bf < b^{2}g \Rightarrow \frac{bf - b^{2}g - bj -2b\sqrt{bgj}}{gb + j + 2\sqrt{bgj}} < 0
\end{equation}
We conclude that all the eigenvalues of $\displaystyle \lim_{\mathcal{I} \rightarrow +\infty} J$ have negative real part, i.e., this equilibrium is linearly stable. Moreover, looking at \eqref{behaviour_jacobian_infinity_immune_part_ii}, we also conclude that
\begin{equation} \label{behaviour_jacobian_infinity_immune_part_iii}
    \lim_{\mathcal{I} \rightarrow +\infty} \det J = + \infty
\end{equation}

Let us consider the homeostasis matrix $H$. As shown before
\begin{equation} \label{not_asymptotic_null_degradation}
    \lim_{\mathcal{I} \rightarrow +\infty} f_{\tau, \tau} = \frac{bf - b^{2}g - bj -2b\sqrt{bgj}}{gb + j + 2\sqrt{bgj}} < 0
\end{equation}
i.e., the system does not present asymptomatic null-degradation homeostasis. Furthermore, looking to the dynamics of $\dot{x}_o$ and considering that for all $\mathcal{I} \in \mathbb{R}^{*}_{+}$, at equilibrium we have $\dot{x}_o = 0$ and so
\begin{equation} \label{limit_product_xiota_xo_infinity_immune}
\begin{split}
     \lim_{\mathcal{I} \rightarrow +\infty} \dot{x}_o = 0 
     &\Rightarrow \lim_{\mathcal{I} \rightarrow +\infty} a \tilde{x}_\sigma \tilde{x}_o - b \tilde{x}_o - l \tilde{x}_\iota \tilde{x}_o + h \tilde{x}_\tau \tilde{x}_o + j = 0 \\ 
     & \Rightarrow \lim_{\mathcal{I} \rightarrow +\infty} - l \tilde{x}_\iota \tilde{x}_o + j = 0 \\
     & \Rightarrow \lim_{\mathcal{I} \rightarrow +\infty}  \tilde{x}_\iota \tilde{x}_o = \frac{j}{l}
\end{split}
\end{equation}

Now we may verify that the system does not exhibit asymptotic structural homeostasis. In fact, applying \eqref{structural_homeostasis_immune} and \eqref{limit_product_xiota_xo_infinity_immune}, we get
\begin{equation} \label{not_asymptotic_structural}
\begin{aligned}
    & \lim_{\mathcal{I} \rightarrow +\infty} f_{o, \sigma}f_{\sigma, \iota} - f_{o, \iota}f_{\sigma, \sigma} = \lim_{\mathcal{I} \rightarrow +\infty} - \tilde{x}_{o} [ae\tilde{x}_\sigma + le\tilde{x}_\iota + le\tilde{x}_o + lf] = -ej \neq 0
\end{aligned}
\end{equation}
i.e., the system does not present asymptotic structural homeostasis. Let's now verify if the system presents asymptotic homeostasis. In fact, by \eqref{not_asymptotic_null_degradation} and \eqref{not_asymptotic_structural}, we conclude that
\[
    \lim_{\mathcal{I} \rightarrow +\infty} \det H = -ej \left( \frac{bf - b^{2}g - bj -2b\sqrt{bgj}}{gb + j + 2\sqrt{bgj}}\right) = ej \left( \frac{b^{2}g + bj + 2b\sqrt{bgj} - bf}{gb + j + 2\sqrt{bgj}}\right) > 0
\]
Observe that $ \displaystyle \lim_{\mathcal{I} \rightarrow +\infty} \det H$ is a finite positive real number. Applying now \eqref{behaviour_jacobian_infinity_immune_part_iii} and the Cramer's Rule, we conclude that
\begin{equation}
    \lim_{\mathcal{I} \rightarrow +\infty} \frac{d\tilde{x}_o}{d\mathcal{I}}(\mathcal{I}) = \lim_{\mathcal{I} \rightarrow +\infty} - \frac{\det H}{\det J} = 0
\end{equation}

Therefore, despite the fact that the system does not present neither asymptotic null degradation or asymptotic structural homeostasis, it still exhibits asymptotic homeostasis.

\section{Intracellular Copper Regulation}

\subsection{Brief Review of Copper Regulation}

Copper is an inorganic element essential to many physiological process, including neurotransmission, gastrointestinal uptake, lactation, transport to the developing brain  and growth. However, its concentration must be tightly regulated, as intracellular copper excess is associated to cellular damage and protein folding disorders \cite{lutsenko07b, kaplan16}.

In addition to cytosolic copper concentration, copper in intramitochondrial space must be also strictly regulated, as it is paramount for the function of copper dependent enzymes, but it may cause oxidative stress in excessive levels \cite{baker17}.

Copper in the external medium enters the cell by CTR1. In the cytosol, copper is rapidly incorporated to glutatione, from where it is ligated to metallochaperones, as ATOX1, CCS and COX17. ATOX1 is associated to the copper secretory pathway, while CCS and COX17 are enrolled in incorporating copper in the mitochondrial enzymes SOD1 and COX \cite{lutsenko07b, kaplan16}.

The ATOX1 protein takes the cytosolic copper to the Cu-ATPases ATP7A and ATP7B, which use ATP to pump copper ions to vesicles of the trans-Golgi network, where copper will be incorporated in Cu-dependent enzymes and secreted. This is called the secretory pathway and it is responsible for decreasing the cytosolic copper concentration. However, when cytosolic copper levels are low, ATP7A and ATP7B take copper from the trans-Golgi network and give it to ATOX1, leading to an increase on the cytosolic copper concentration~\cite{yu17}.

The functions governed by copper homeostasis are primarily executed by the copper-transporting ATPases known as ATP7A and ATP7B. ATP7A is a transmembrane protein located throughout the body, except for the liver, with two essential roles in copper homeostasis: transporting copper across cell membranes in both directions (regulating absorption of copper only in the small intestines, and excreting excessive intracellular copper, in all tissues) aiming therefore at the maintenance of intracellular copper concentrations (both cytosolic and mitochondrial); and participating as a cofactor in the activating mechanisms of copper-dependant enzymes, critical for the structure and function of bone, skin, hair, blood vessels, and the nervous system \cite{vosko02, scheiber13}. On the other hand, the ATP7B transmembrane protein is located primarily in liver cells, but also in the brain, and bears similar tasks: regulating intracellular copper concentrations by releasing copper into bile and plasma, and co-activating copper-dependant enzymes in the Golgi apparatus \cite{lutsenko07b, scheiber13}. 

Expanding briefly on the physiological implications of defective copper regulation, anomalies in the ATP7B gene generate a sole disorder known as Wilson disease (WD), in which dysfunctional ATP7B proteins implicate WD carriers to accumulate abnormal levels of copper in the liver and in the brain. As a result, clinical features comprise neurological, hepatic, psychiatric and skeletal abnormalities, as well as renal tubular dysfunction and hemolytic anemia. The prognosis in WD is generally favorable given that current therapeutic approaches prevent or attenuate most of the symptoms. Its chronic nature, however, implies that treatment interruption results in potentially fatal liver damage \cite{kaler08}. 

Differently, variations in the ATP7A gene result in dysfunctional ATP7A proteins that cause three separate illnesses: Menkes disease, a severe early-onset neurodegenerative condition in which carriers usually die by 3 years of age \cite{kaler94a}; occipital horn syndrome, a connective disorder with typical skeleton deformations which is also clinically resembling to Menkes disease, while less aggressive in its neurological manifestation \cite{kaler94b}; and a recently found distal motor neuropathy, marked by frequent onset at adulthood and with no apparent signs of copper metabolic abnormalities, although still poorly studied \cite{kennerson10, yi12}.

\begin{figure}[!ht]
\centering
\includegraphics[scale = 0.45,trim=2cm 1.5cm 2cm 0cm,clip=true]{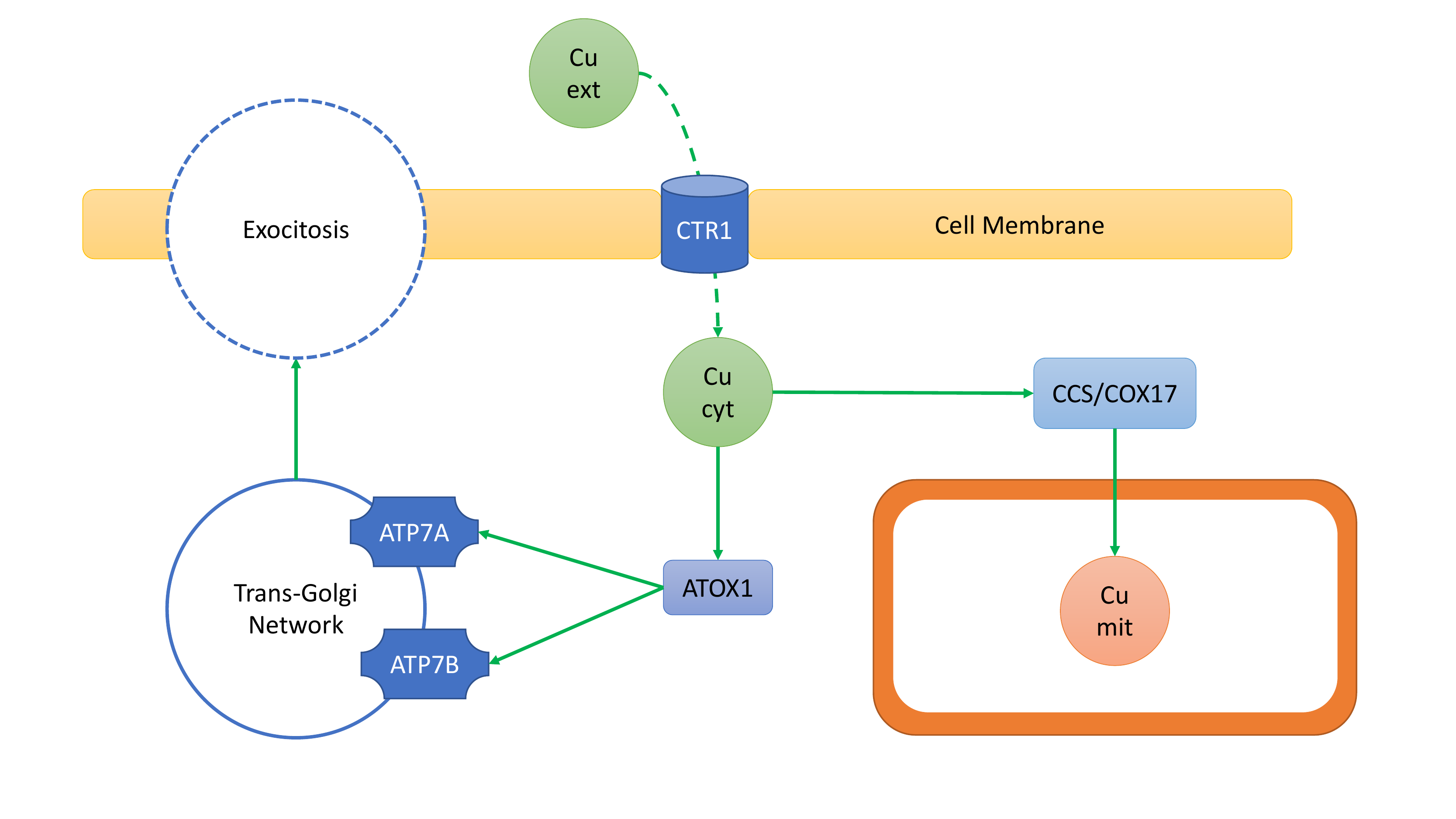}
\caption{\label{copper_simple_model}
Simplified model of intracellular copper regulation. 
Here, Cu ext: extracellular copper; Cu cyt: cytosolic copper; Cu mit: mitochondrial copper.}
\end{figure}

\subsection{Mathematical Model}

A simplified version of the intracellular copper regulation mechanism described above can be obtained by considering the concentration of copper in three environments: extracellular copper (Cu ext:), cytosolic copper (Cu cyt), mitochondrial copper (Cu mit).  
The dynamics of copper concentration on these environments is governed by its interaction with three metallochaperones: ATOX1, CCS and COX17. 
This interaction dynamics is represented by the diagram of Figure~\ref{copper_simple_model}.

\begin{figure}[!ht]
\centering
\begin{tabular}{c@{}c}
\includegraphics[scale = 0.3,trim=5cm 3cm 4cm 3cm,clip=true]{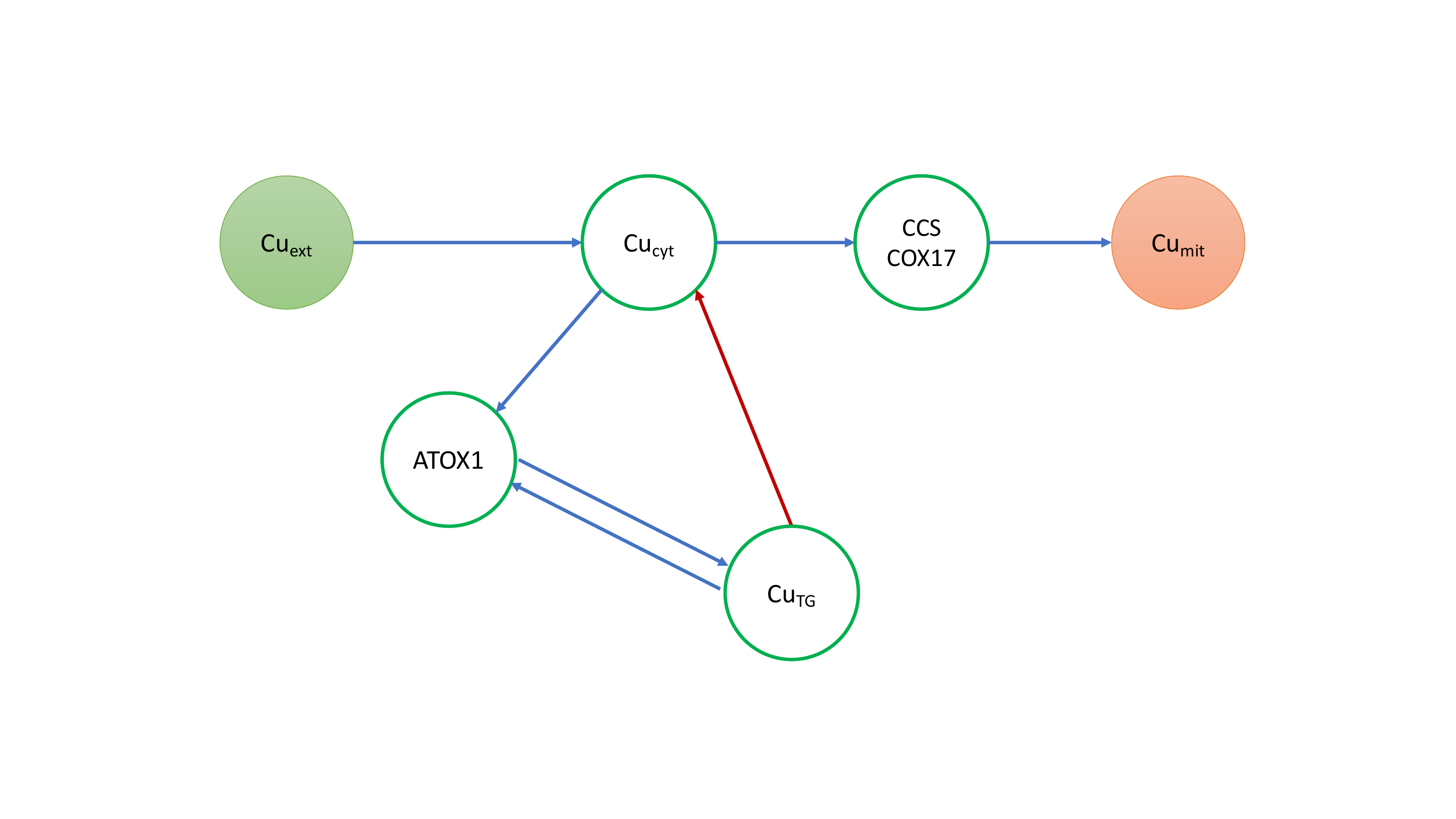} &
\includegraphics[scale = 0.3,trim=5cm 3cm 5cm 3cm,clip=true]{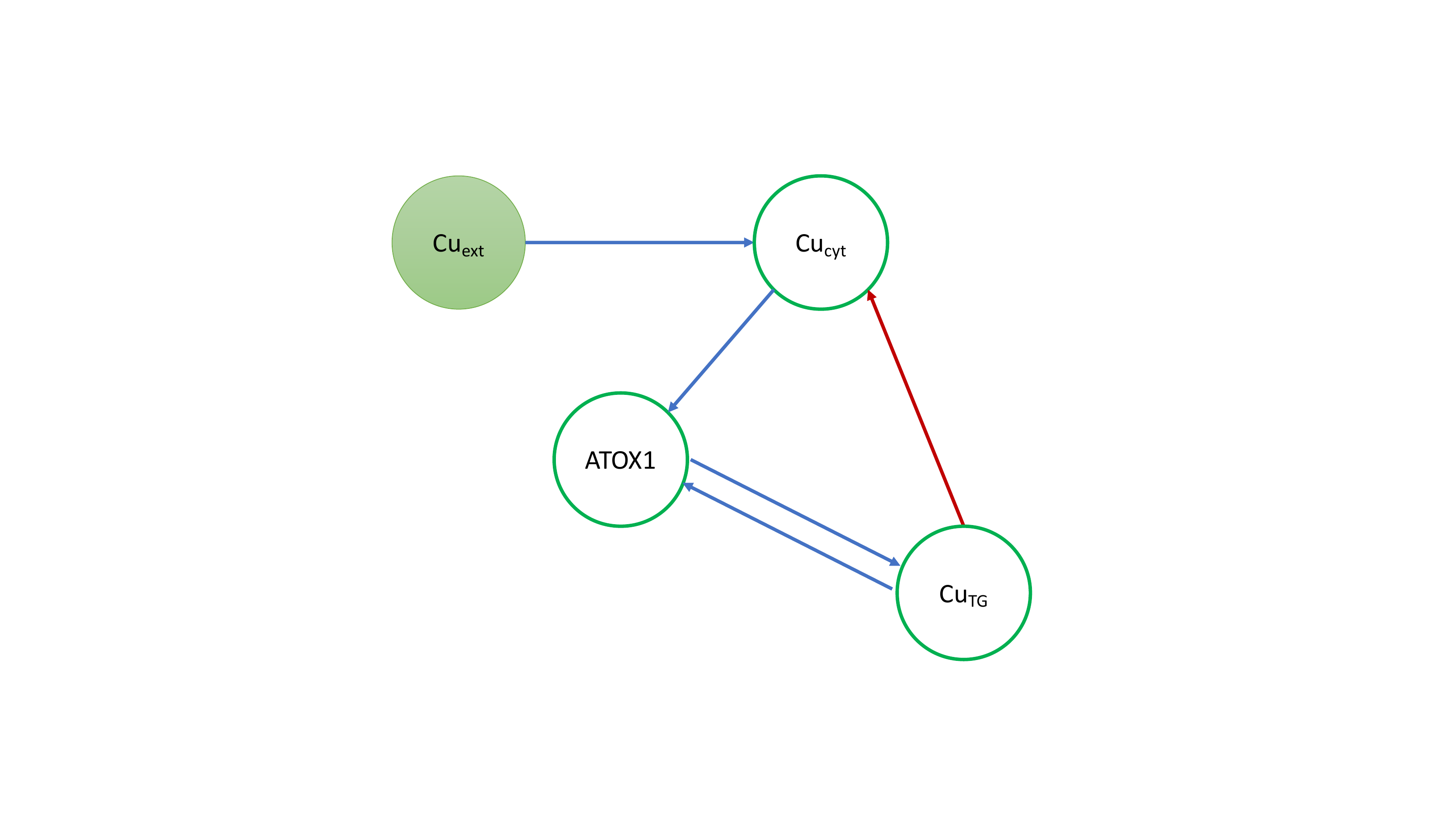} \\
(A) & (B)
\end{tabular}
\caption{\label{fig_copper_network}
Input-output network for the intracellular copper regulation model. 
Blue arrows indicate positive stimulus (activation) and red arrows indicate negative stimulus (inhibition). (A) Full network. (B) Core network.
}
\end{figure}

We can abstract this model by the inout-output network shown in Figure \ref{fig_copper_network}(A). Here the extracellular copper concentration $[\text{Cu}_{\text{ext}}]$ is the input node and mitochondrial copper concentration $[\text{Cu}_{\text{mit}}]$ is the output node.
The input parameter $\mathcal{I}$ represents the abundance of extracellular copper.
As observed before, in order to verify that $[\text{Cu}_{\text{mit}}]$ is homeostatic, it is enough to verify that $[\text{Cu}_{\text{cyt}}]$ is homeostatic. 
Hence, we can further simplify the input-output network of Figure \ref{fig_copper_network}(A) to its core network shown in Figure \ref{fig_copper_network}(B). To facilitate notation, let's represent the concentrations of $\text{Cu}_{\text{ext}}$, $\text{Cu}_{\text{cyt}}$, $\text{ATOX1}$ and $\text{Cu}_{\text{TG}}$, respectively, as $x_{\iota}$, $x_{o}$, $x_{\tau}$ and $x_{\rho}$. Then the dynamical system associated to the network in Figure \ref{fig_copper_network}(B) becomes
\begin{equation} \label{simplified_eqs_newvar}
\begin{aligned}
& \dot{x}_\iota = \mathcal{I} - k_0 x_\iota \\
& \dot{x}_\tau = f k_1 x_o - k_3x_\tau - w_2 \frac{x_\tau (x_\rho - x_\tau)} {1 + x_\tau} \\
& \dot{x}_\rho = g k_3 x_\tau + w_2 \frac{x_\tau (x_\rho - x_\tau)}{1 + x_\tau} - k_4 x_\rho \\
& \dot{x}_o = \frac{k_0}{N} x_\iota - k_1 x_o (1 + w_1 x_o) + k_2 G(x_\rho)
\end{aligned}
\end{equation}

\noindent
Here, the constants $N$, $w_1$, $w_2$, $k_0$, $k_1$, $k_2$, $k_3$, $k_4$ are positive parameters, $f,g \in (0,1]$ and $G$ and $H$ are quadratic Hill Functions (for $x\geq 0$):
\begin{equation}
 G(x)=\frac{1}{1+x^2}-1 \qquad\text{and}\qquad
 H(x)=\frac{x}{1+x}
\end{equation}
Notice that this system is represented by the abstract network shown in Figure \ref{four_nodes_example_feedback}(A).

\subsection{Infinitesimal Homeostasis}

The jacobian matrix $J$ of \eqref{simplified_eqs_newvar} at an equilibrium point is 
\begin{equation}
\begin{aligned}
    J & = \begin{pmatrix}
    f_{\iota, x_{\iota}} & 0 & 0 & 0 \\
    0 & f_{\tau, x_{\tau}} & f_{\tau, x_{\rho}} & f_{\tau, x_{o}} \\
    0 & f_{\rho, x_{\tau}} & f_{\rho, x_{\rho}} & 0 \\
    f_{o, x_{\iota}} & 0 & f_{o, x_{\rho}}  & f_{o, x_{o}}
    \end{pmatrix} 
\end{aligned}
\end{equation}    
    
\begin{equation}
\begin{aligned}
    J & = \begin{pmatrix}
    -k_0 & 0 & 0 & 0 \\
    0 & - k_3 + w_2 \frac{x_\tau^2 + 2 x_\tau - x_\rho}{(1+x_\tau)^2} & - w_2 \frac{x_\tau}{1+x_\tau} & f k_1 \\
    0 & g k_3 - w_2 \frac{x_\tau^2 + 2 x_\tau - x_\rho}{(1+x_\tau)^2} & - k_4 + w_2 \frac{x_\tau}{1+x_\tau} & 0  \\
    \frac{k_0}{N} & 0 & k_2 G_{x_{\rho}}(x_\rho) & -k_1(1+2 w_1 x_o)
    \end{pmatrix}
\end{aligned}
\end{equation}

Note that, for $\mathcal{I}=0$, the point $(0,0,0,0)$ is a solution and the jacobian at $(0,0,0,0)$ is (recall that $G_{x_{\rho}}(0)=0$)
\begin{equation}
  J = \begin{pmatrix}
  -k_0 & 0 & 0 & 0 \\
  0 & -k_3 & 0 & fk_{1} \\
  0 & g k_3 & - k_4 & 0 \\
  \frac{k_0}{N} & 0 & 0  & - k_1
  \end{pmatrix}
\end{equation}
and so $(0,0,0,0)$ is always stable.

On the other hand, analysing the abstract network shown in Figure \ref{four_nodes_example_feedback}, we conclude that:

\begin{equation} \label{det_H_copper_model}
\begin{aligned}
    \det (H) & = \begin{vmatrix}
    f_{\iota, x_{\iota}} & 0 & 0 & -1 \\
    0 & f_{\tau, x_{\tau}} & f_{\tau, x_{\rho}} & 0 \\
    0 & f_{\rho, x_{\tau}} & f_{\rho, x_{\rho}} & 0 \\
    f_{o, x_{\iota}} & 0 & f_{o, x_{\rho}}  & 0
    \end{vmatrix} \Rightarrow \det (H) = f_{o, x_{\iota}} \cdot (f_{\tau, x_{\tau}}f_{\rho, x_{\rho}} - f_{\tau, x_{\rho}}f_{\rho, x_{\tau}}) \\
 \Rightarrow \det(H) & =\frac{k_0}{N} \left(k_4 \left(k_3+w_2\frac{x_\rho-x_\tau^2-2 x_\tau}{(1+x_\tau)^2}\right) + k_3 w_2 (g-1) \left(\frac{x_\tau}{1+x_\tau} \right)\right)
\end{aligned}
\end{equation}
For $\mathcal{I}=0$ we have that
\begin{equation}
 \det(H)=\frac{k_0 k_3 k_4}{N} \neq 0
\end{equation}

Moreover, by equation \eqref{det_H_copper_model}, the abstract network supports Haldane and Appendage homeostasis. However, regarding the intracellular copper regulation system, by equation \eqref{det_H_copper_model}, we have:

\begin{equation}
    f_{o, x_{\iota}} = \frac{k_{0}}{N} \neq 0
\end{equation}
and therefore if the system exhibits homeostasis, it exhibits appendage homeostasis. In the graph below we show a simulation of this system in XPP which exhibits homeostasis.

\begin{figure}[H]
\centering
\includegraphics[scale = 0.35]{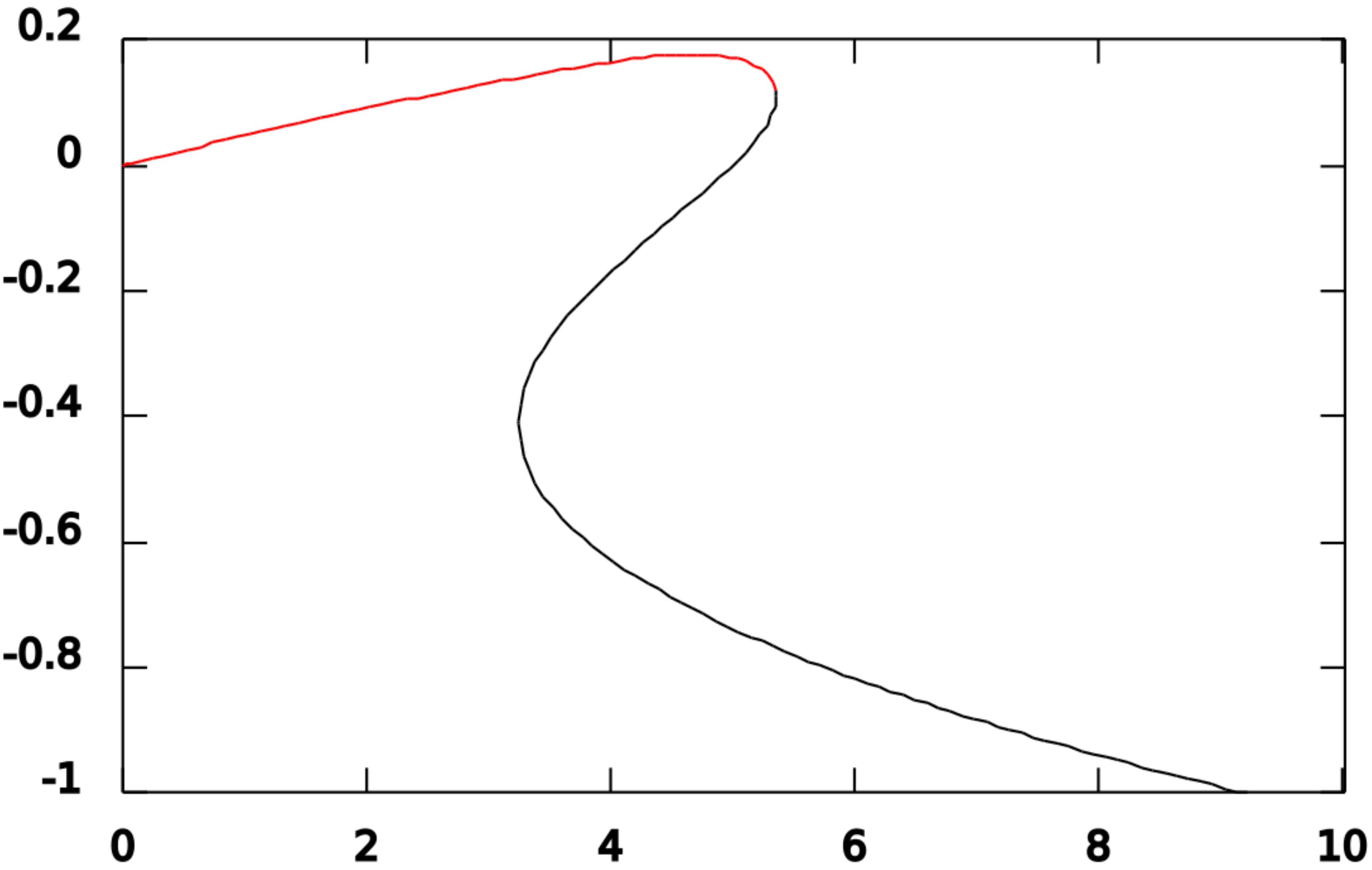}
\renewcommand{\figurename}{Figure}
\caption{Figure generated by \textsc{Xpp-Auto} for the input-output map $x_0$ ($y$-axis) as function of $\mathcal{I}$ ($x$-axis), named $J$ in the picture. In this case the point of infinitesimal homeostasis is around $\mathcal{I}_0=4.7$.The red line indicated that the equilibrium is stable and the black line indicates that the equilibrium is unstable; the exchange of stability occurs around $\mathcal{I}=5.2$. The parameter values are: $N=10$, $f=0.5$, $g=0.05$, $w_1=1$, $w_2=0.5$ , $k_0=10$, $k_1=2$, $k_2=1$, $k_3=0.5$, $k_4=1$.}
\label{fig_diagram}
\end{figure}

From a biological perspective, the classification of homeostasis as appendage homeostasis may provide useful information about the studied system, as we shall see in the following subsections.

\subsection{Normal Form of the Input-Output Function}

Another important qualitative feature of the system is the normal form of the input-output function around the homeostasis point, i.e., if the system supports chair homeostasis for some choice of parameters or not. This is important because, as noted by Golubitsky et al. \cite{gs17}, simple homeostasis is qualitatively different from chair homeostasis.

The graph shown in Figure \ref{fig_diagram} suggests that for the simulated set of parameters the system presented simple homeostasis. However, it is important to analytically study this question, as parameters in biological systems are hard to determine and may present great variations among individuals.

We shall than apply the fact that the system exhibits appendage homeostasis to simplify the computation of $\displaystyle \frac{d^{2} \tilde{x}_o}{d \mathcal{I}^2}$. Firstly, let's represent the equilibrium points of the system as $(\tilde{x}_{\iota}, \tilde{x}_{\tau}, \tilde{x}_{\rho}, \tilde{x}_{o})$.

Remember that, according to equation \eqref{det_H_copper_model}, the determinant of the homeostasis matrix of the corresponding abstract network is:

\begin{equation*}
    \det H = f_{o, x_{\iota}} \cdot (f_{\tau, x_{\tau}}f_{\rho, x_{\rho}} - f_{\tau, x_{\rho}}f_{\rho, x_{\tau}})
\end{equation*}

Considering that the system exhibits appendage homeostasis, as noted by Golubtisky et al., to determine the normal form of the input-output function around the homeostasis point we may evaluate the derivative of the appendage sub-network as the system presents appendage homeostasis. Denominating $\det H_{1} = f_{\tau, x_{\tau}}f_{\rho, x_{\rho}} - f_{\tau, x_{\rho}}f_{\rho, x_{\tau}}$, we must evaluate $\displaystyle \frac{d \det H_{1}}{d \mathcal{I}}$. By the chain rule, we got:

\begin{equation} \label{derivada_determinante_subrede_apendicular_primeira_parte}
    \frac{d \det H_{1}}{d \mathcal{I}} = \frac{\partial \det H_{1}}{\partial \mathcal{I}} + \frac{\partial \det H_{1}}{\partial \tilde{x}_{\iota}} \cdot \frac{d \tilde{x}_\iota}{d \mathcal{I}} + \frac{\partial \det H_{1}}{\partial \tilde{x}_{\tau}} \cdot \frac{d \tilde{x}_\tau}{d \mathcal{I}} + \frac{\partial \det H_{1}}{\partial \tilde{x}_{\rho}} \cdot \frac{d \tilde{x}_\rho}{d \mathcal{I}} + \frac{\partial \det H_{1}}{\partial \tilde{x}_{o}} \cdot \frac{d \tilde{x}_o}{d \mathcal{I}}
\end{equation}

As we are evaluating this at the homeostasis point, than $\displaystyle \frac{d \tilde{x}_o}{d \mathcal{I}} = 0$. Furthermore, the expression of $\det H_{1}$ does not explicitly depend on $\mathcal{I}$ or $\tilde{x}_{\iota}$. Therefore, we may simplify \eqref{derivada_determinante_subrede_apendicular_primeira_parte}, obtaining:

\begin{equation} \label{general_formula_derivative_appendage}
    \frac{d \det H_{1}}{d \mathcal{I}} = \frac{\partial \det H_{1}}{\partial \tilde{x}_{\tau}} \cdot \frac{d \tilde{x}_\tau}{d \mathcal{I}} + \frac{\partial \det H_{1}}{\partial \tilde{x}_{\rho}} \cdot \frac{d \tilde{x}_\rho}{d \mathcal{I}}
\end{equation}

Now we can use the explicit formula for $\det H_{1}$ used in \eqref{det_H_copper_model}:

\begin{equation*}
    \det H_{1} = k_4 \left(k_3 - w_2\frac{\tilde{x}_\tau^2 + 2 \tilde{x}_\tau - \tilde{x}_\rho}{(1+\tilde{x}_\tau)^2}\right) + k_3 w_2 (g-1) \left(\frac{\tilde{x}_\tau}{1+\tilde{x}_\tau} \right) 
\end{equation*}
to compute the partial derivatives:

\begin{equation} \label{explicit_partial_derivatives_appendage_block}
\begin{aligned}
    & \frac{\partial \det H_{1}}{\partial \tilde{x}_{\tau}} = \frac{(g-1)k_{3}w_{2}}{(1+\tilde{x}_\tau)^2} - \frac{2k_{4}w_{2}(\tilde{x}_\rho + 1)}{(1+\tilde{x}_\tau)^3} \\
    & \frac{\partial \det H_{1}}{\partial \tilde{x}_{\rho}} =\frac{k_4 w_{2}}{(1+\tilde{x}_\tau)^2}
\end{aligned}
\end{equation}

We must now compute $\displaystyle \frac{d \tilde{x}_\tau}{d \mathcal{I}}$ and $\displaystyle \frac{d \tilde{x}_\rho}{d \mathcal{I}}$. In order to perform this, we shall use a strategy analogous to the one used to obtain the homeostasis matrix. In fact, remember that, as shown by Golubitsky et al. \cite{gs17}, considering $J$ as the Jacobian at the homeostasis point and that $f_{\iota, \mathcal{I}} = 1$, than the following linear system is satisfied:

\begin{equation} \label{cramer_segunda_derivada_primeira_parte}
    J \begin{pmatrix}[1.5] \displaystyle \frac{d \tilde{x}_\iota}{d \mathcal{I}} \\ \displaystyle \frac{d \tilde{x}_\tau}{d \mathcal{I}} \\ \displaystyle \frac{d \tilde{x}_\rho}{d \mathcal{I}} \\ 
    \displaystyle \frac{d \tilde{x}_o}{d \mathcal{I}}
    \end{pmatrix} = \begin{pmatrix}
    - 1 \\ 0 \\ 0 \\ 0
    \end{pmatrix}
\end{equation}

As the equilibrium must be linearly stable, than $\det J \neq 0$, and therefore we may apply Cramer's rule to compute $\displaystyle \frac{d \tilde{x}_\tau}{d \mathcal{I}}$ and $\displaystyle \frac{d \tilde{x}_\rho}{d \mathcal{I}}$. Therefore, we can write:

\begin{equation} \label{cramer_segunda_derivada_segunda_parte}
    \frac{d \tilde{x}_\tau}{d \mathcal{I}} = \frac{\det H_{\tau}}{\det J} \textrm{ and } \frac{d \tilde{x}_\rho}{d \mathcal{I}} = \frac{\det H_{\rho}}{\det J}
\end{equation}
where

\begin{equation} \label{primeira_parte_derivadas_xtauexro_I}
    \det H_{\tau} = \begin{vmatrix} 
    f_{\iota, x_{\iota}} & -1 & 0 & 0 \\
    0 & 0 & f_{\tau, x_{\rho}} & f_{\tau, x_{o}} \\
    0 & 0 & f_{\rho, x_{\rho}} & 0 \\
    f_{o, x_{\iota}} & 0 & f_{o, x_{\rho}}  & f_{o, x_{o}}
    \end{vmatrix}
    \textrm{ and } \det H_{\rho} = \begin{vmatrix} 
    f_{\iota, x_{\iota}} & 0 & -1 & 0 \\
    0 & f_{\tau, x_{\tau}} & 0 & f_{\tau, x_{o}} \\
    0 & f_{\rho, x_{\tau}} & 0 & 0 \\
    f_{o, x_{\iota}} & 0 & 0  & f_{o, x_{o}}
    \end{vmatrix}
\end{equation}

By \eqref{primeira_parte_derivadas_xtauexro_I}, we conclude that:

\begin{equation} \label{segunda_parte_derivadas_xtauexro_I}
    \det H_{\tau} = - f_{o, x_{\iota}} f_{\tau, x_{o}} f_{\rho, x_{\rho}}
    \textrm{ and } \det H_{\rho} = f_{o, x_{\iota}} f_{\tau, x_{o}} f_{\rho, x_{\tau}}
\end{equation}

Applying \eqref{cramer_segunda_derivada_segunda_parte} and \eqref{segunda_parte_derivadas_xtauexro_I} to \eqref{general_formula_derivative_appendage}, we obtain:

\begin{equation}
    \frac{d \det H_{1}}{d \mathcal{I}} = \frac{f_{o, x_{\iota}} f_{\tau, x_{o}}}{\det J} \left( f_{\rho, x_{\tau}} \frac{\partial \det H_{1}}{\partial \tilde{x}_{\rho}} - f_{\rho, x_{\rho}} \frac{\partial \det H_{1}}{\partial \tilde{x}_{\tau}}\right)
\end{equation}

We have already proved that $\det J \neq 0$ and $f_{o, x_{\iota}} \neq 0$. Moreover, as seen above, the feedback loop $ o \rightarrow \rho \rightarrow \tau \rightarrow o$ must be a negative feedback loop, which means that $f_{\tau, x_{o}} \neq 0$. Therefore, in order to the system present chair homeostasis, we must have:

\begin{equation} \label{necessary_condition_for_chair_homeostasis}
    f_{\rho, x_{\tau}} \frac{\partial \det H_{1}}{\partial \tilde{x}_{\rho}} - f_{\rho, x_{\rho}} \frac{\partial \det H_{1}}{\partial \tilde{x}_{\tau}} = 0
\end{equation}

Remember that, in the studied system we have:

\begin{equation} \label{explicit_formula_partial_derivatives_appendage}
    f_{\rho, x_{\tau}} = g k_3 - w_2 \frac{\tilde{x}_\tau^2 + 2 \tilde{x}_\tau - \tilde{x}_\rho}{(1+\tilde{x}_\tau)^2} \textrm{ and } f_{\rho, x_{\rho}} = - k_4 + w_2 \frac{\tilde{x}_\tau}{1+\tilde{x}_\tau}
\end{equation}

Applying \eqref{explicit_partial_derivatives_appendage_block} and \eqref{explicit_formula_partial_derivatives_appendage} to \eqref{necessary_condition_for_chair_homeostasis}, we obtain:

\begin{equation} \label{first_part_no_chair_in_copper}
\setlength{\jot}{15pt}
\begin{aligned}
& \left[ g k_3 - w_2 \frac{\tilde{x}_\tau^2 + 2 \tilde{x}_\tau - \tilde{x}_\rho}{(1+\tilde{x}_\tau)^2} \right] \cdot \frac{k_4 w_{2}}{(1+\tilde{x}_\tau)^2} + \left(k_4 - w_2 \frac{\tilde{x}_\tau}{1+\tilde{x}_\tau} \right) \cdot \left[ \frac{(g-1)k_{3}w_{2}}{(1+\tilde{x}_\tau)^2} - \frac{2k_{4}w_{2}(\tilde{x}_\rho + 1)}{(1+\tilde{x}_\tau)^3} \right] = 0 \\
& \Rightarrow \frac{(2g-1)k_{3}k_{4}w_{2}}{(1+\tilde{x}_\tau)^2} - \frac{k_{4}w_{2}^{2}(\tilde{x}_\tau^2 + 2 \tilde{x}_\tau - \tilde{x}_\rho)}{(1+\tilde{x}_\tau)^4} -\frac{2k_{4}^{2}w_{2}(\tilde{x}_\rho + 1)}{(1+\tilde{x}_\tau)^3} - \frac{(g-1)k_{3}w_{2}^{2}\tilde{x}_{\tau}}{(1+\tilde{x}_\tau)^3} + \frac{2k_{4}w_{2}^{2}(\tilde{x}_\rho + 1)\tilde{x}_{\tau}}{(1+\tilde{x}_\tau)^4} = 0 \\
& \Rightarrow \frac{(2g-1)k_{3}k_{4}w_{2}}{(1+\tilde{x}_\tau)^2} - \frac{w_{2}}{(1+\tilde{x}_\tau)^{2}}\left[ \frac{(g-1)k_{3}w_{2}\tilde{x}_{\tau}}{(1+\tilde{x}_\tau)} - \frac{k_{4}w_{2}(\tilde{x}_\tau^2 + 2 \tilde{x}_\tau - \tilde{x}_\rho)}{(1+\tilde{x}_\tau)^2} \right] - \frac{2k_{4}w_{2}^{2}(\tilde{x}_\tau^2 + 2 \tilde{x}_\tau - \tilde{x}_\rho)}{(1+\tilde{x}_\tau)^4} \\
& -\frac{2k_{4}^{2}w_{2}(\tilde{x}_\rho + 1)}{(1+\tilde{x}_\tau)^3} + \frac{2k_{4}w_{2}^{2}(\tilde{x}_\rho + 1)\tilde{x}_{\tau}}{(1+\tilde{x}_\tau)^4} = 0
\end{aligned}
\end{equation}

Now remind that the system present appendage homeostasis and by \eqref{det_H_copper_model}, we obtain:

\begin{equation} \label{explicit_condition_appendage_homeostasis}
\setlength{\jot}{15pt}
\begin{aligned}
    \det H_{1} = 0 \Leftrightarrow & \frac{(g-1)k_{3}w_{2}\tilde{x}_{\tau}}{(1+\tilde{x}_\tau)} - \frac{k_{4}w_{2}(\tilde{x}_\tau^2 + 2 \tilde{x}_\tau - \tilde{x}_\rho)}{(1+\tilde{x}_\tau)^2} + k_{3}k_{4} = 0 \\
   \Leftrightarrow & \frac{(g-1)k_{3}w_{2}\tilde{x}_{\tau}}{(1+\tilde{x}_\tau)} - \frac{k_{4}w_{2}(\tilde{x}_\tau^2 + 2 \tilde{x}_\tau - \tilde{x}_\rho)}{(1+\tilde{x}_\tau)^2} = - k_{3}k_{4}
\end{aligned}
\end{equation}

Applying \eqref{explicit_condition_appendage_homeostasis} to \eqref{first_part_no_chair_in_copper}, we obtain:

\begin{equation} \label{second_part_no_chair_in_copper}
\setlength{\jot}{15pt}
\begin{aligned}
    & \frac{2gk_{3}k_{4}w_{2}}{(1+\tilde{x}_\tau)^2} - \frac{2k_{4}w_{2}^{2}(\tilde{x}_\tau^2 + 2 \tilde{x}_\tau - \tilde{x}_\rho)}{(1+\tilde{x}_\tau)^4} -\frac{2k_{4}^{2}w_{2}(\tilde{x}_\rho + 1)}{(1+\tilde{x}_\tau)^3} + \frac{2k_{4}w_{2}^{2}(\tilde{x}_\rho + 1)\tilde{x}_{\tau}}{(1+\tilde{x}_\tau)^4} = 0 \\
    & \Rightarrow \frac{2k_{4}w_{2}}{(1+\tilde{x}_\tau)^2}\left[ gk_{3} - \frac{w_{2}(\tilde{x}_\tau^2 + 2 \tilde{x}_\tau - \tilde{x}_\rho)}{(1+\tilde{x}_\tau)^2} \right] + \frac{2k_{4}w_{2}(1 + \tilde{x}_{\rho})}{(1 + \tilde{x}_{\tau})^{3}} \left[ -k_{4} + \frac{w_{2}\tilde{x}_{\tau}}{(1+\tilde{x}_\tau)}\right] = 0
\end{aligned}
\end{equation}

Applying now \eqref{explicit_formula_partial_derivatives_appendage} to \eqref{second_part_no_chair_in_copper}:

\begin{equation} \label{third_part_no_chair_in_copper}
\begin{aligned}
    \frac{2k_{4}w_{2}}{(1+\tilde{x}_\tau)^2} f_{\rho, x_{\tau}} + \frac{2k_{4}w_{2}(1 + \tilde{x}_{\rho})}{(1 + \tilde{x}_{\tau})^{3}} f_{\rho, x_{\rho}} = 0 \Rightarrow \frac{2k_{4}w_{2}}{(1+\tilde{x}_\tau)^3} \left[ (1 + \tilde{x}_\tau)f_{\rho, x_{\tau}} + (1 + \tilde{x}_\rho) f_{\rho, x_{\rho}} \right] = 0
\end{aligned}
\end{equation}

Analysing equation \eqref{third_part_no_chair_in_copper}, it is easy to see that $\displaystyle \frac{2k_{4}w_{2}}{(1+\tilde{x}_\tau)^3} \neq 0$, and therefore in order to the system exhibit chair homeostasis, we must have:

\begin{equation} \label{fourth_part_no_chair_in_copper}
    (1 + \tilde{x}_\tau)f_{\rho, x_{\tau}} + (1 + \tilde{x}_\rho) f_{\rho, x_{\rho}} = 0
\end{equation}

As we are analysing the system in its point of appendage homeostasis, this means that the following equations must be simultaneously satisfied:

\begin{equation}
\begin{aligned}
    (1 + \tilde{x}_\tau)f_{\rho, x_{\tau}} + (1 + \tilde{x}_\rho) f_{\rho, x_{\rho}} & = 0 \\
    - f_{\tau, x_{\rho}}f_{\rho, x_{\tau}} + f_{\tau, x_{\tau}} f_{\rho, x_{\rho}} & = 0
\end{aligned}
\end{equation}

If we analyse these equations as an homogeneous linear system in variables $f_{\rho, x_{\tau}}$ and $f_{\rho, x_{\rho}}$ and remembering that $f_{\rho, x_{\tau}} \neq 0$ as $o \rightarrow \rho \rightarrow \tau \rightarrow o$ is a negative feedback loop, than we conclude that

\begin{equation} \label{fifth_part_no_chair_in_copper}
\begin{aligned}
    \begin{vmatrix} (1 + \tilde{x}_\tau) & (1 + \tilde{x}_\rho) \\ - f_{\tau, x_{\rho}} & f_{\tau, x_{\tau}} \end{vmatrix} = 0 \Rightarrow (1 + \tilde{x}_\tau)f_{\tau, x_{\tau}} + (1 + \tilde{x}_\rho) f_{\tau, x_{\rho}} = 0
\end{aligned}
\end{equation}

Remember that in the studied system we have:

\begin{equation} \label{explicit_derivatives_f_tau}
\begin{aligned}
f_{\tau, x_{\tau}} = - k_3 + w_2 \frac{\tilde{x}_\tau^2 + 2 \tilde{x}_\tau - \tilde{x}_\rho}{(1+\tilde{x}_\tau)^2} \textrm{ and } f_{\tau, \tilde{x}_{\rho}} = - w_2 \frac{\tilde{x}_\tau}{(1+\tilde{x}_\tau)}
\end{aligned}
\end{equation}

We may substitute \eqref{explicit_derivatives_f_tau} in \eqref{fifth_part_no_chair_in_copper}, obtaining:

\begin{equation} \label{sixth_part_no_chair_in_copper}
\begin{aligned}
    & (1 + \tilde{x}_\tau)f_{\tau, x_{\tau}} + (1 + \tilde{x}_\rho) f_{\tau, x_{\rho}} = 0 \\ &\Leftrightarrow \;
    - k_3 (1 + \tilde{x}_\tau) + w_2 \frac{\tilde{x}_\tau^2 + 2 \tilde{x}_\tau - \tilde{x}_\rho}{(1+\tilde{x}_\tau)} - w_2 \frac{\tilde{x}_\tau (1 + \tilde{x}_\rho)}{(1+\tilde{x}_\tau)} = 0 \\
    & \Leftrightarrow \;
    w_2 \frac{\tilde{x}_\tau^2 + 2 \tilde{x}_\tau - \tilde{x}_\rho}{(1+\tilde{x}_\tau)} - w_2 \frac{\tilde{x}_\tau (1 + \tilde{x}_\rho)}{(1+\tilde{x}_\tau)} = k_3 (1 + \tilde{x}_\tau)
\end{aligned}
\end{equation}

Applying now \eqref{explicit_formula_partial_derivatives_appendage} to \eqref{fourth_part_no_chair_in_copper}, we got:

\begin{equation} \label{seventh_part_no_chair_in_copper}
\begin{aligned}
    & (1 + \tilde{x}_\tau)f_{\rho, x_{\tau}} + (1 + \tilde{x}_\rho) f_{\rho, x_{\rho}} = 0 \\
    & \Leftrightarrow  \;
    g k_{3}(1 + \tilde{x}_\tau) - \left[w_2 \frac{\tilde{x}_\tau^2 + 2 \tilde{x}_\tau - \tilde{x}_\rho}{(1+\tilde{x}_\tau)} - w_2 \frac{\tilde{x}_\tau (1 + \tilde{x}_\rho)}{(1+\tilde{x}_\tau)} \right] - k_{4}(1 + \tilde{x}_\rho) = 0
\end{aligned}
\end{equation}

Finally, we can apply \eqref{sixth_part_no_chair_in_copper} to \eqref{seventh_part_no_chair_in_copper} in order to get:

\begin{equation} \label{final_part_no_chair_in_copper}
\begin{aligned}
    gk_{3}(1 + \tilde{x}_\tau) - k_{3}(1 + \tilde{x}_\tau) - k_{4}(1 + \tilde{x}_\rho) = 0 \Leftrightarrow (g-1)k_{3}(1 + \tilde{x}_\tau) - k_{4}(1 + \tilde{x}_\rho) = 0
\end{aligned}
\end{equation}

From the model, it is reasonable to consider $1 + \tilde{x}_\rho > 0$ and $1 + \tilde{x}_\tau > 0$ and therefore, as $0 < g \leq 1$, than \eqref{final_part_no_chair_in_copper} is a contradiction, which implies that a point of appendage homeostasis of the system is a point of simple homeostasis.

\bibliographystyle{abbrv}
\bibliography{refs}

\begin{thebibliography}{10}

\bibitem{abbas14}
A.~K. Abbas, A.~H. Lichtman, and S.~Pillai.
\newblock {\em Cellular and Molecular Immunology E-book}.
\newblock Elsevier Health Sciences, 2014.

\bibitem{ang13}
J.~Ang and D.~R. McMillen.
\newblock Physical constraints on biological integral control design for
  homeostasis and sensory adaptation.
\newblock {\em Biophys. J.}, 104(2):505--515, 2013.

\bibitem{antoneli18}
F.~Antoneli, M.~Golubitsky, and I.~Stewart.
\newblock Homeostasis in a feed forward loop gene regulatory motif.
\newblock {\em J. Theor. Biol.}, 445:103--109, 2018.

\bibitem{bacchetta18}
R.~Bacchetta, F.~Barzaghi, and M.-G. Roncarolo.
\newblock From {IPEX} syndrome to {FOXP3} mutation: a lesson on immune
  dysregulation.
\newblock {\em Ann. N.Y. Acad. Sci.}, 1417(1):5--22, 2018.

\bibitem{baker17}
Z.~N. Baker, P.~A. Cobine, and S.~C. Leary.
\newblock The mitochondrion: a central architect of copper homeostasis.
\newblock {\em Metallomics}, 9(11):1501--1512, 2017.

\bibitem{best09}
J.~A. Best, H.~F. Nijhout, and M.~C. Reed.
\newblock Homeostatic mechanisms in dopamine synthesis and release: a
  mathematical model.
\newblock {\em Theor. Biol. Med. Modell.}, 6(1):21, 2009.

\bibitem{gs06}
M.~Golubitsky and I.~Stewart.
\newblock Nonlinear dynamics of networks: the groupoid formalism.
\newblock {\em Bull. Amer. Math. Soc.}, 43(3):305--364, 2006.

\bibitem{gs17}
M.~Golubitsky and I.~Stewart.
\newblock Homeostasis, singularities, and networks.
\newblock {\em J. Math. Biol.}, 74(1-2):387--407, 2017.

\bibitem{gs18}
M.~Golubitsky and I.~Stewart.
\newblock Homeostasis with multiple inputs.
\newblock {\em SIAM J. Appl. Dynam. Sys.}, 17(2):1816--1832, 2018.

\bibitem{gsahy20}
M.~Golubitsky, I.~Stewart, F.~Antoneli, Z.~Huang, and Y.~Y. Wang.
\newblock Input-output networks, singularity theory, and homeostasis.
\newblock In O.~Junge, S.~Ober-Blobaum, K.~Padburg-Gehle, G.~Froyland, and
  O.~Schütze, editors, {\em Advances in Dynamics, Optimization and
  Computation}, pages 36--65. Springer Cham, 2020.

\bibitem{yangyang20}
M.~Golubitsky and Y.~Wang.
\newblock Infinitesimal homeostasis in three-node input-output networks.
\newblock {\em J. Math. Biol.}, 80:1163--1185, 2020.

\bibitem{huang22}
Z.~Huang and M.~Golubitsky.
\newblock Classification of infinitesimal homeostasis in four-node input-output
  networks.
\newblock {\em Preprint}, pages 1--20, 2022.

\bibitem{kaler94a}
S.~G. Kaler.
\newblock Menkes disease.
\newblock {\em Adv. Pediatr.}, 41:263--304, 1994.

\bibitem{kaler08}
S.~G. Kaler.
\newblock {\em Cecil Textbook of Medicine}, chapter 230: ``Wilson disease''.
\newblock Saunders, Philadelphia, 23rd edition, 2008.

\bibitem{kaler94b}
S.~G. Kaler, L.~K. Gallo, V.~K. Proud, et~al.
\newblock Occipital horn syndrome and a mild {M}enkes phenotype associated with
  splice site mutations at the {MNK} locus.
\newblock {\em Nat. Genet.}, 8:195--202, 1994.

\bibitem{kaplan16}
J.~H. Kaplan and E.~B. Maryon.
\newblock How mammalian cells acquire copper: an essential but potentially
  toxic metal.
\newblock {\em Biophys. J.}, 110(1):7--13, 2016.

\bibitem{kennerson10}
M.~L. Kennerson, G.~A. Nicholson, S.~G. Kaler, et~al.
\newblock Missense mutations in the copper transporter gene atp7a cause
  x-linked distal hereditary motor neuropathy.
\newblock {\em Am J Hum Genet}, 86:343–352, 2010.

\bibitem{khailaie13}
S.~Khailaie, F.~Bahrami, M.~Janahmadi, P.~Milanez-Almeida, J.~Huehn, and
  M.~Meyer-Hermann.
\newblock A mathematical model of immune activation with a unified self-nonself
  concept.
\newblock {\em Front Immunol}, 4:474, 2013.

\bibitem{lloyd13}
A.~C. Lloyd.
\newblock The regulation of cell size.
\newblock {\em Cell}, 154:1194, 2013.

\bibitem{lutsenko07b}
S.~Lutsenko, N.~L. Barnes, M.~Y. Bartee, and O.~Y. Dmitriev.
\newblock Function and regulation of human copper-transporting {ATP}ases.
\newblock {\em Physiol Rev}, 87(3):1011--1046, 2007.

\bibitem{ma09}
W.~Ma, A.~Trusina, H.~El-Samad, W.~A. Lim, and C.~Tang.
\newblock Defining network topologies that can achieve biochemical adaptation.
\newblock {\em Cell}, 138(4):760--773, 2009.

\bibitem{nijhout14}
H.~F. Nijhout, J.~Best, and M.~C. Reed.
\newblock Escape from homeostasis.
\newblock {\em Math. Biosci.}, 257:104--110, 2014.

\bibitem{nbr18}
H.~F. Nijhout, J.~Best, and M.~C. Reed.
\newblock Systems biology of robustness and homeostatic mechanisms.
\newblock {\em WIREs Syst. Biol. Med.}, page e1440, 2018.

\bibitem{nbr15}
H.~F. Nijhout, J.~A. Best, and M.~C. Reed.
\newblock Using mathematical models to understand metabolism, genes and
  disease.
\newblock {\em BMC Biol.}, 13:79, 2015.

\bibitem{nr14}
H.~F. Nijhout and M.~C. Reed.
\newblock Homeostasis and dynamic stability of the phenotype link robustness
  and plasticity.
\newblock {\em Integr. Comp. Biol.}, 54(2):264--75, 2014.

\bibitem{nrbu04}
H.~F. Nijhout, M.~C. Reed, P.~Budu, and C.~M. Ulrich.
\newblock A mathematical model of the folate cycle: new insights into folate
  homeostasis.
\newblock {\em J. Biol. Chem.}, 279:55008--55016, 2004.

\bibitem{reed17}
M.~Reed, J.~Best, M.~Golubitsky, I.~Stewart, and H.~F. Nijhout.
\newblock Analysis of homeostatic mechanisms in biochemical networks.
\newblock {\em Bull. Math. Biol.}, 79(11):2534--2557, 2017.

\bibitem{scheiber13}
I.~Scheiber, R.~Dringen, and J.~F.~B. Mercer.
\newblock {\em Interrelations between essential metal ions and human diseases},
  volume~13 of {\em Metal Ions in Life Sciences}, chapter 11: ``Copper: Effects
  of Deficiency and Overload''.
\newblock Springer-Verlag, 2013.

\bibitem{sok18}
S.~P.~M. Sok, D.~Ori, N.~H. Nagoor, and T.~Kawai.
\newblock Sensing self and non-self {DNA} by innate immune receptors and their
  signaling pathways.
\newblock {\em Crit. Rev. Immunol.}, 38(4), 2018.

\bibitem{tang16}
Z.~F. Tang and D.~R. McMillen.
\newblock Design principles for the analysis and construction of robustly
  homeostatic biological networks.
\newblock {\em J. Theor. Biol.}, 408:274--289, 2016.

\bibitem{vosko02}
I.~Voskoboinik and J.~Camakaris.
\newblock Menkes copper-translocating {P}-type {ATP}ase {(ATP7A)}: biochemical
  and cell biology properties, and role in {M}enkes disease.
\newblock {\em J. Bioenerg. Biomembr.}, 34:363--71, 2002.

\bibitem{wang20}
Y.~Wang, Z.~Huang, F.~Antoneli, and M.~Golubitsky.
\newblock The structure of infinitesimal homeostasis in input-output networks.
\newblock {\em J. Math. Biol.}, 82:62, 2021.

\bibitem{wyatt99}
J.~K. Wyatt, A.~R.-D. Cecco, C.~A. Czeisler, and D.-J. Dijk.
\newblock Circadian temperature and melatonin rhythms, sleep, and
  neurobehavioral function in humans living on a 20-h day.
\newblock {\em Am. J. Physiol.}, 277:1152--1163, 1999.

\bibitem{yi12}
L.~Yi, A.~Donsante, M.~L. Kennerson, et~al.
\newblock Altered intra-cellular localization and valosin-containing protein
  (p97 {VCP}) interaction underlie {ATP7A}-related distal motor neuropathy.
\newblock {\em Hum Mol Genet}, 21:1794–1807, 2012.

\bibitem{yu17}
C.~H. Yu, N.~V. Dolgova, and O.~Y. Dmitriev.
\newblock Dynamics of the metal binding domains and regulation of the human
  copper transporters {ATP7B} and {ATP7A}.
\newblock {\em IUBMB Life}, 69(4):226--235, 2017.

\end{thebibliography}

\end{document}